\documentclass[
aps,pra,
reprint,
superscriptaddress,
twocolumn, 
nofootinbib,
floats,floatfix
]{revtex4-2}

\usepackage{pifont}
\usepackage{amssymb}
\usepackage{amsthm}
\usepackage{amsmath}
\usepackage{bbm}
\usepackage{graphicx}
\usepackage{xcolor}
\usepackage[english]{babel}
\usepackage{csquotes}
\usepackage{braket}
\usepackage{mathtools}
\usepackage{enumitem}
\usepackage{IEEEtrantools}
\usepackage{relsize}
\usepackage{placeins}
\usepackage{comment}
\usepackage[normalem]{ulem}
\usepackage{systeme}
\usepackage{physics}
\usepackage{listings}
\usepackage{xcolor}
\usepackage{mathtools}
\usepackage{tcolorbox}
\usepackage[ruled,vlined]{algorithm2e}
\usepackage{lipsum}

\usepackage{hyperref}

\renewcommand{\vec}[1]{\boldsymbol{#1}}

\newcommand{\Ebb}{\mathbb{E}}

\newcommand{\cP}{\mathcal{P}}
\newcommand{\cM}{\mathcal{M}}

\newcommand{\OC}{\mathcal{O}}

\newcommand{\supp}{\text{supp}}

\renewcommand{\geq}{\geqslant}
\renewcommand{\leq}{\leqslant}

\newcommand{\poly}{\operatorname{poly}}

\newcommand{\wt}{\mathrm{wt}}
\newcommand{\1}{\mathbb{I}}

\newtheorem{theorem}{Theorem}
\newtheorem{lemma}[theorem]{Lemma}
\newtheorem{proposition}[theorem]{Proposition}
\newtheorem{definition}[theorem]{Definition}

\newtheorem{corollary}[theorem]{Corollary}

\renewcommand{\arraystretch}{1.2}

\newcommand{\coef}{{w}} 
\newcommand{\nstate}{{A}}
\newcommand{\pbf}{{q_{\rm bf}}} 
\newcommand{\weight}{{w}}

\newcommand{\owt}{{h}}

\makeatletter
\renewcommand\onecolumngrid{
\do@columngrid{one}{\@ne}
\def\set@footnotewidth{\onecolumngrid}
\def\footnoterule{\kern-6pt\hrule width 1.5in\kern6pt}
}
\renewcommand\twocolumngrid{
        \def\footnoterule{
        \dimen@\skip\footins\divide\dimen@\thr@@
        \kern-\dimen@\hrule width.5in\kern\dimen@}
        \do@columngrid{mlt}{\tw@}
}
\makeatother

\makeatletter
\newif\ifeqcontrib@this
\newif\ifeqcontrib@any
\eqcontrib@thisfalse
\eqcontrib@anyfalse

\newcommand{\eqcontrib}{\global\eqcontrib@thistrue\global\eqcontrib@anytrue}

\newcommand{\eqcontribmark}{\textsuperscript{\ensuremath{,\S}}}

\newcommand{\eqcontrib@maybe}{%
  \ifeqcontrib@this
    \global\eqcontrib@thisfalse
    \eqcontribmark
  \fi
}

\newcommand{\printEqContrib}{%
  \ifeqcontrib@any
    \begingroup
      \renewcommand{\thefootnote}{\ensuremath{\S}}%
      \footnotetext{These authors contributed equally to this work.}%
    \endgroup
  \fi
}

\def\doauthor#1#2#3{%
  \ignorespaces#1\unskip\@listcomma
  \begingroup
    #3%
  \@if@empty{#2}{\endgroup{}{}}{\endgroup{\comma@space}{}\frontmatter@footnote{#2}}%
  \eqcontrib@maybe
  \space \@listand
}%
\makeatother

\usepackage{hyperref}
\hypersetup{
  pdfborder={0 0 0}
}

\begin{document}

\title{Thermal State Simulation with Pauli and Majorana Propagation}

\author{Manuel S. Rudolph\eqcontrib}
\email{manuel.rudolph@epfl.ch}
\affiliation{Institute of Physics, Ecole Polytechnique F\'{e}d\'{e}rale de Lausanne (EPFL),   Lausanne, Switzerland}
\affiliation{Centre for Quantum Science and Engineering, Ecole Polytechnique F\'{e}d\'{e}rale de Lausanne (EPFL),   Lausanne, Switzerland}

\author{Armando Angrisani\eqcontrib}
\email{armando.angrisani@epfl.ch}
\affiliation{Institute of Physics, Ecole Polytechnique F\'{e}d\'{e}rale de Lausanne (EPFL),   Lausanne, Switzerland}
\affiliation{Centre for Quantum Science and Engineering, Ecole Polytechnique F\'{e}d\'{e}rale de Lausanne (EPFL),   Lausanne, Switzerland}

\author{Andrew Wright}
\affiliation{Institute of Physics, Ecole Polytechnique F\'{e}d\'{e}rale de Lausanne (EPFL),   Lausanne, Switzerland}
\affiliation{Centre for Quantum Science and Engineering, Ecole Polytechnique F\'{e}d\'{e}rale de Lausanne (EPFL),   Lausanne, Switzerland}

\author{Iwo Sanderski}
\affiliation{Institute of Physics, Ecole Polytechnique F\'{e}d\'{e}rale de Lausanne (EPFL),   Lausanne, Switzerland}
\affiliation{Centre for Quantum Science and Engineering, Ecole Polytechnique F\'{e}d\'{e}rale de Lausanne (EPFL),   Lausanne, Switzerland}

\author{Ricard Puig}
\email{ricard.puigivalls@epfl.ch}
\affiliation{Institute of Physics, Ecole Polytechnique F\'{e}d\'{e}rale de Lausanne (EPFL),   Lausanne, Switzerland}
\affiliation{Centre for Quantum Science and Engineering, Ecole Polytechnique F\'{e}d\'{e}rale de Lausanne (EPFL),   Lausanne, Switzerland}

\author{Zo\"{e} Holmes}
\affiliation{Institute of Physics, Ecole Polytechnique F\'{e}d\'{e}rale de Lausanne (EPFL),   Lausanne, Switzerland}
\affiliation{Centre for Quantum Science and Engineering, Ecole Polytechnique F\'{e}d\'{e}rale de Lausanne (EPFL),   Lausanne, Switzerland}
\affiliation{Algorithmiq Ltd, Kanavakatu 3 C, FI-00160 Helsinki, Finland}

\date{\today}

\begin{abstract}
We introduce a propagation-based approach to thermal state simulation by adapting Pauli and Majorana propagation to imaginary-time evolution in the Schr\"odinger picture. Our key observation is that high-temperature states can be sparse in the Pauli or Majorana bases, approaching the identity at infinite temperature. By formulating imaginary-time evolution directly in these operator bases and evolving from the maximally mixed state, we access a continuum of temperatures where the state remains efficiently representable. We provide analytic guarantees for small-coefficient truncation and Pauli-weight (Majorana-length) truncation strategies by quantifying the error growth and the impact of backflow. Large-scale numerics on the 1D $J_1$--$J_2$ model (energies) and the triangular-lattice Hubbard model (static correlations) validate efficiency at high temperatures.

\end{abstract}

\maketitle

\printEqContrib

\section{Introduction}

Understanding quantum matter at finite temperatures is central to material science, condensed matter and quantum chemistry. From a fundamental perspective, thermal fluctuations compete with quantum coherence and can change phases of matter, drive phase transitions, and determine observables such as susceptibilities and transport coefficients. From an applications perspective, many target problems, including reaction rates and free energies are inherently finite-temperature. Classically probing these properties requires efficient methods for approximating thermal states.

Propagation methods, most commonly Pauli~\cite{rall2019simulation,aharonov2022polynomial, beguvsic2023simulating, fontana2023classical, shao2023simulating, rudolph2023classical, schuster2024polynomial, angrisani2024classically, gonzalez2024pauli, lerch2024efficient, cirstoiu2024fourier, angrisani2025simulating, fuller2025improved, rudolph2025pauli, angrisani2025simulating, teng2025leveraging} and Majorana~\cite{miller2025simulation,alam2025fermionic,alam2025programmable,d2025majorana, facelli2026fast} propagation, are relative newcomers to the classical toolbox for simulating quantum systems.  At their core, propagation methods approximate the evolution of a quantum operator via a truncated path integral.
Due to Pauli propagation's initial conceptualization for simulating real-time dynamics~\cite{rudolph2023classical,beguvsic2023fast}, the family of propagation algorithms as a whole are currently primarily understood as tools for precisely that. Both in the continuous-time~\cite{loizeau2025quantum} or discrete quantum circuit formulations~\cite{aharonov2022polynomial, beguvsic2023simulating, fontana2023classical, shao2023simulating, rudolph2023classical, schuster2024polynomial, angrisani2024classically, gonzalez2024pauli, lerch2024efficient, cirstoiu2024fourier, angrisani2025simulating, fuller2025improved, rudolph2025pauli, angrisani2025simulating, teng2025leveraging}. Furthermore, they are almost exclusively presented as methods for estimating expectation values in the Heisenberg picture, i.e., for the backward evolution of observables that are typically sparser in the Pauli or Majorana basis than states. 

In this work, we challenge this perspective and propose a propagation algorithm for the simulation of finite-temperature quantum states in the Schrödinger picture. While zero-temperature pure states are always dense in the Pauli or Majorana basis, consisting of at least $2^n$ operators for an $n$-site system, high-temperature states can be exceptionally sparse. In the limit of infinite temperature, the maximally-mixed state is representable by a single operator: the identity operator. By defining the action of imaginary-time evolution in the Pauli and Majorana basis, and starting to evolve the maximally mixed state, we gain access to a range of temperatures at which we can efficiently store the quantum state in memory.

We analyze two complementary truncation strategies, small-coefficient truncation and Pauli-weight (and Majorana-length) truncation. We derive upper bounds showing how the resulting approximation error decreases as the truncation threshold is increased, guaranteeing that propagation methods can run efficiently at high temperatures. We complement these results with large-scale numerical experiments: for the 1D $J_1$--$J_2$ model we study energy estimates across system sizes and temperatures, and for the Fermi-Hubbard model on a triangular lattice we demonstrate direct computation of finite-temperature static correlation functions from the propagated thermal state. Both our analytic and numerical results indicate that simulating high temperatures is efficient but low temperatures remain challenging. 

Our algorithm has several natural applications. Firstly, it provides a route to estimating free energies, since once a thermal state (or its partition function) is accessible, thermodynamic quantities follow with minimal additional work. Secondly, as we can easily combine thermal preparation with subsequent time evolution, we can, for example, probe finite-temperature corrections to infinite-temperature dynamical correlation functions. Thirdly, the ability to represent thermal states compactly suggests applications to Gibbs sampling and generative modeling, where one aims to efficiently draw samples from thermal distributions or learn compact generative surrogates for finite-temperature data~\cite{Minervini2026Strong}. More broadly, imaginary time propagation methods retain two practical advantages: they make no assumptions about lattice topology and they interface naturally with quantum hardware.

\begin{figure*}
    \centering
    \includegraphics[width=0.9\linewidth]{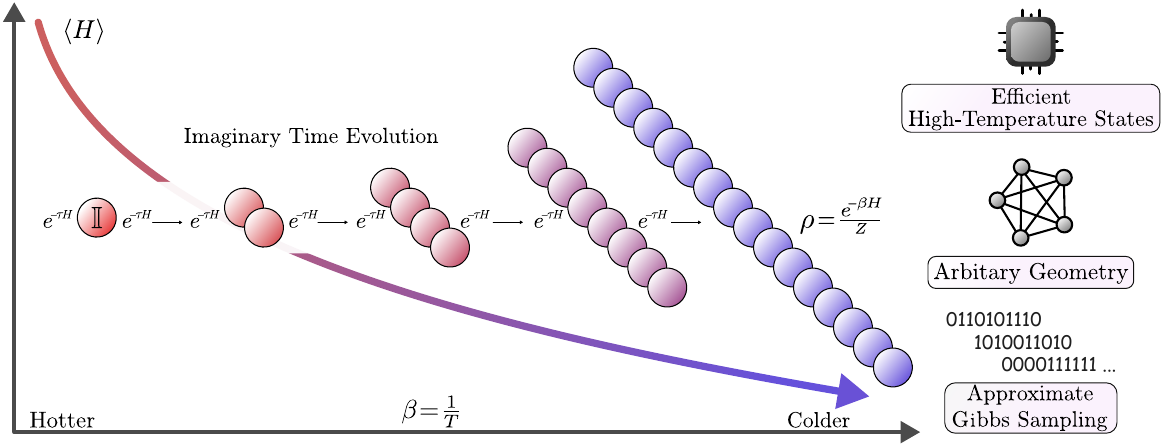}
    \caption{\textbf{Schematic depiction of thermal state preparation via propagation methods.} The process begins at infinite temperature ($\beta=0$) with the maximally mixed state, represented by the identity operator $\mathbb{I}$. We then apply a sequence of imaginary-time gates (propagators $e^{-\tau H}$) to evolve the operator. As the system cools (corresponding to larger $\beta$), the state becomes less sparse and the number of operators (either Paulis or Majoranas) that need to be propagated grow. This approach allows us to efficiently store and manipulate high-temperature states.}
    \label{fig:fig1}
\end{figure*}

\section{Framework}

In this paper, we develop methods to simulate the thermal (Gibbs) state at inverse temperature $\beta$,
\begin{equation}\label{eq:thermalstate}
    \rho_{\beta}
    = \frac{e^{-\beta H}}{\Tr\!\left( e^{-\beta H}\right)}
    = \frac{e^{-\frac{\beta}{2} H}\,\mathbb{I}\,e^{-\frac{\beta}{2} H}}{\Tr\!\left( e^{-\beta H}\right)},
\end{equation}
where $H$ is the Hamiltonian of interest. As shown in Eq.~\eqref{eq:thermalstate}, the numerator can be viewed as evolving the identity operator by the non-unitary imaginary-time propagator $e^{-\frac{\beta}{2}H}$, i.e., as a ``rotation'' by an imaginary angle $\tau=-i\tfrac\beta2$.
Pauli~\cite{rall2019simulation,aharonov2022polynomial, beguvsic2023simulating, fontana2023classical, shao2023simulating, rudolph2023classical, schuster2024polynomial, angrisani2024classically, gonzalez2024pauli, lerch2024efficient, cirstoiu2024fourier, angrisani2025simulating, fuller2025improved, rudolph2025pauli, angrisani2025simulating, teng2025leveraging} and Majorana~\cite{miller2025simulation,alam2025fermionic,alam2025programmable,d2025majorana, facelli2026fast} propagation methods can be used to compute the action of this imaginary-time evolution on $\mathbb{I}$. We will focus initially on describing the method for Pauli propagation but the Majorana generalization is entirely analogous. 

\medskip

\paragraph*{Pauli Propagation.} The first step to simulating imaginary time evolution via Pauli propagation is to express the Hamiltonian as a sum of (local) Pauli terms,
\begin{equation}
    H = \sum_{m=1}^{M} h_m, 
    \qquad 
    h_m = \lambda_m P_m,
    \quad 
    P_m\in\mathcal{P}_n,
\end{equation}
and apply a Trotter approximation of the imaginary evolution operator, for example, a first-order Trotter approximation,
\begin{equation}
    e^{-\beta H} \;\approx\; \left(\prod_{m=1}^{M} e^{- \tau\lambda_m P_m}\right)^{L} \, \, \, , L := \beta /  \tau 
\end{equation}
where the Trotter error associated with this approximation scales as $\OC(C \beta \tau)$ with $C:=\sum_{m<m'} \bigl\|[h_m,h_{m'}]\bigr\|$. 
This approximation decomposes the evolution for imaginary time $\beta$ into a product of $L \times M$ local short $\tau$ ``imaginary-time gates'' 
each of which can be applied iteratively to a Pauli expansion, starting from $\mathbb{I}$.

The action of an imaginary time gate generated by some Pauli operator $P$ on another Pauli operator $Q$ is
\begin{equation}\label{eq:branching-rule}
    e^{-\frac{\tau}{2}P}\,Q\,e^{-\frac{\tau}{2}P} =
        \begin{cases}
            Q, & [P,Q]\neq 0,\\[4pt]
            \cosh(\tau)\,Q - \sinh(\tau)\,P Q, & [P,Q]=0,
        \end{cases}
\end{equation}
where for Pauli strings $P,Q$, the product $PQ$ is also a Pauli string, up to a phase.
Note that contrary to the real-time case, branching occurs when $P$ and $Q$ commute. Also, all coefficients remain real due to the product of commuting Pauli strings producing an even number of imaginary factors. 

After imaginary-time Pauli propagation, where the gates are iteratively applied to the initial identity operator, we obtained a collection of Pauli operators $P$'s weighted by coefficients $\alpha_P$'s, such that
\begin{equation}
    e^{-\frac{\beta}{2}H}\,\mathbb{I}\,e^{-\frac{\beta}{2}H}
    \;\approx\;
    \sum_{P\in\mathcal{P}_n} \alpha_P\, P,
\end{equation}
and the thermal state can be written as
\begin{equation}\label{eq:rho_beta}
    \rho_\beta
    \;\approx\;
    \frac{\sum_{P\in\mathcal{P}_n}\alpha_P P}{\Tr\!\left(\sum_{P\in\mathcal{P}_n}\alpha_P P\right)}
    \;=\;
    \frac{\sum_{P\in\mathcal{P}_n}\alpha_P P}{2^n\,\alpha_{\mathbb{I}}},
\end{equation}
since $\Tr(P)=0$ for all non-identity Paulis and $\Tr(\mathbb{I})=2^n$. In particular, the partition function $\Tr(e^{-\beta H})$ is proportional to the identity coefficient $\alpha_{\mathbb{I}}$.

In general, the number of Pauli terms generated by repeated application of these gates grows exponentially with the number of gates. Pauli propagation methods usually remain practical (or even formally polynomially scaling)
when combined with truncation strategies that limit the growth of the Pauli expansion to the most relevant contributions. Two natural choices are \emph{small-coefficient truncation}, where one discards terms whose amplitudes $|\alpha_P|$ are smaller
than some cutoff value and \emph{Pauli-weight truncation}, where one removes Pauli strings whose weight, i.e., the number of non-identity single-qubit factors in $P$, exceeds a prescribed cutoff. The latter exploits the fact that local observables are sensitive primarily to low-weight contributions and high-weight Pauli strings, have a tendency to increase in weight rather than
\emph{backflow} to low weights under subsequent propagation.

\medskip
\paragraph*{Randomized product formulas.}
For the purposes of our analysis (or if reducing Trotter error is a primary interest), we will also consider an alternative to deterministic Trotterization based on the quantum stochastic drift protocol (qDRIFT)~\cite{Campbell2019random}. In this setting one defines the 1-norm of the Hamiltonian coefficients $\Lambda := \sum_{m=1}^{M} |\lambda_m|$ and samples indices $m$ i.i.d.\ from the distribution $p_m = |\lambda_m|/\Lambda$. A qDRIFT approximation to the imaginary-time propagator is then obtained by applying a sequence of randomly chosen local gates,
\begin{equation}
    e^{-\beta H} \;\approx\; \prod_{\ell=1}^{L_{\rm qD}} \exp\!\left(- \tau\, s_{m_\ell} P_{m_\ell}\right), \, \, L_{\rm qD} := \Lambda \beta / \tau
    \label{eq:qdrift_evolution}
\end{equation}
where each $m_\ell\sim p$ is drawn independently and $s_{m_\ell}:=\mathrm{sign}(\lambda_{m_\ell})$.  
Intuitively, qDRIFT replaces $L$ repetitions of the fixed ordering over all $M$ terms by a stochastic walk over $L_{\rm qD}$ Hamiltonian terms. The qDRIFT simulation error is expected to scale as $\mathcal{O}(\Lambda \beta \tau)$ which can be less than the 1st order Trotter error when $\Lambda < C = \sum_{m<m'} \bigl\|[h_m,h_{m'}]\bigr\|$. 

Within Pauli propagation, taking a qDRIFT approach to imaginary time simulation simply means that at each step one updates the current operator expansion by the adjoint action of the sampled gate. Averaging over many sampled sequences yields an unbiased estimator of observables under the approximated imaginary-time evolution, while the per-step cost remains identical to the Trotter case. In the context of this paper the qDRIFT approach to imaginary time simulation is largely an analytic tool that helps us to intuitively understand the probability of backflow from high to low weight Paulis in our algorithm and thereby derive guarantees based on weight-truncation. However, we also foresee that it could become the go-to method for the Monte-Carlo implementations of our algorithm that will ultimately be necessary to overcome memory limitations and utilize GPU capabilities~\cite{rudolph2025pauli}. \\

\paragraph*{Majorana propagation.} In the context of fermionic systems, it is more natural to work in the Majorana operator basis. For a system of $n$ fermionic modes $\{a_j\}_{j=1}^N$, the Majorana operator basis is defined by
\begin{equation}
    m_{2j-1} := a_j + a_j^\dagger, 
    \, \, 
    m_{2j} := -i(a_j - a_j^\dagger),
    \qquad j\in\{1,\dots,N\},
\end{equation}
which satisfy $\{m_p,m_q\}=2\delta_{pq}$ and $m_p^\dagger=m_p$.
A \emph{Majorana monomial} is indexed by a binary vector $\mathbf b=(b_1,\dots,b_{2N})\in\{0,1\}^{2N}$ and defined as
\begin{equation}
    M_{\mathbf b} \;:=\; i^{r_{\mathbf b}} \, m_1^{b_1} m_2^{b_2}\cdots m_{2N}^{b_{2N}},
\end{equation}
with the phase choice $r_{\mathbf b}\in\{0,1\}$ ensuring $M_{\mathbf b}$ is Hermitian (equivalently, $r_{\mathbf b}=0$ when $\|\mathbf b\|_1\equiv 0,1 \!\!\pmod{4}$ and $r_{\mathbf b}=1$ otherwise). The \emph{length} of a monomial is
\begin{equation}
    |M_{\mathbf b}| \;:=\; \|\mathbf b\|_1 \;=\; \sum_{j=1}^{2N} b_j,
\end{equation}
and the set $\Gamma:=\{M_{\mathbf b}:\mathbf b\in\{0,1\}^{2N}\}$ forms an operator basis. Hence any operator admits a unique expansion
\begin{equation}
    O \;=\; \sum_{\mathbf b\in\{0,1\}^{2N}} \alpha_{\mathbf b}\, M_{\mathbf b},
\end{equation}
with real coefficients $\alpha_{\mathbf b}\in\mathbb{R}$ whenever $O$ is Hermitian.

To simulate imaginary-time evolution in this representation, we expand the Hamiltonian in the same basis,
\begin{equation}
    H \;=\; \sum_{m=1}^{M} h_m,
    \qquad
    h_m = \lambda_m M_{\mathbf b_m},
    \qquad
    M_{\mathbf b_m}\in\Gamma,
\end{equation}
and approximate $e^{-\beta H}$ by a product formula as in the Pauli case, yielding local imaginary-time gates
$G_m := e^{-\Delta\beta\,\lambda_m M_{\mathbf b_m}}$.
Because any two monomials either commute or anticommute and satisfy $M_{\mathbf b}^2=\mathbb{I}$, the update rule on a basis element
$M_{\mathbf a}$ takes the same form as Eq.~\eqref{eq:branching-rule}, namely $    e^{-\frac{\tau}{2}M_{\mathbf b}}\,M_{\mathbf a}\,e^{-\frac{\tau}{2}M_{\mathbf b}}$
\begin{equation}\label{eq:majorana_evolution}
    =
    \begin{cases}
        M_{\mathbf a}, & [M_{\mathbf b},M_{\mathbf a}]\neq 0,\\[4pt]
        \cosh(\tau)\,M_{\mathbf a} - \sinh(\tau)\,M_{\mathbf b}M_{\mathbf a}, & [M_{\mathbf b},M_{\mathbf a}]=0.
    \end{cases}
\end{equation}
In particular, branching occurs when $M_{\mathbf b}$ and $M_{\mathbf a}$ commute. Moreover,
the product $M_{\mathbf b}M_{\mathbf a}$ is again Hermitian (up to an overall sign), so the propagated coefficients remain real.
Finally, truncation can be implemented by discarding monomials with small coefficient prefactors as well as those whose length exceeds a threshold $\ell_{\max}$ (Majorana-length truncation),
in analogy with Pauli-weight truncation. \\

\paragraph*{Relation to other methods.}
At a high level, our approach is a series treatment of imaginary-time evolution in \emph{operator space}. Conventional finite-temperature Monte Carlo methods~\cite{alhassid2001quantum, chang2004quantum, van2006quantum, bulgac2008quantum, militzer2015development, gubernatis2016quantum} 
also exploit a path-integral viewpoint 
to estimate observables without explicitly constructing $\rho_\beta$. In contrast, propagation tracks how basis operators (Pauli strings or Majorana monomials) are transformed under successive imaginary-time gates. The resulting tradeoff replaces the sign-problem barrier~\cite{troyer2005computational} with an operator-growth/branching barrier.

Our truncation rules are closely connected to high-temperature~\cite{baker1967high, gaunt1970low, georges1991expand} and cluster-expansion ideas~\cite{sanchez1984generalized, blum2004mixed,wu2016cluster}. Small-coefficient truncation in the case of small rotation angles mirrors a high-temperature series expansion in which higher-order contributions are suppressed and can be discarded. Pauli-weight (or Majorana-length) truncation plays the role of a cluster-size cutoff in the sense that local observables are primarily sensitive to low-support operators, and contributions from large-support terms are generally suppressed. In this sense, our methods are analogous to a cluster-expansion-like approximation, but organized by operator algebra rather than by explicit enumeration of connected graph components.

Our methods also complement the rich literature on quantum algorithms for preparing or sampling from Gibbs states~\cite{temme2011quantum, poulin2009sampling, chowdhury2017quantum, motta2020determining, rouze2024optimal, consiglio2024variational, lin2025dissipative, cubitt2023dissipative,puig2024variational}, and provide a practical classical baseline for delineating the temperatures and system sizes at which quantum thermal-state preparation may yield a quantum advantage.

Finally, there is a structural analogy to tensor-network thermal states~\cite{verstraete2004matrix, zwolak2004mixed, wolf2008area, stoudenmire2010minimally}. Methods based on purification or MPO/PEPO imaginary-time evolution can also apply local imaginary-time gates followed by truncation to control complexity. The main difference is-- as usually-- the representation and truncation metric: spatial factorization and finite bond dimension versus sparse representation and finite support.
This complementary perspective highlights when propagation may be advantageous: it is topology-agnostic, allows directly reading off the Pauli or Majorana decomposition of the state,
and can interface directly with quantum hardware for hybrid workflows~\cite{lerch2024efficient,fuller2025improved}.

\section{Analytic guarantees}
In this section, we provide guarantees for the efficiency of simulation in low-$\beta$ (i.e., high-temperature) regimes. While our results here are framed in the context of Pauli propagation, everything carries directly over to Majorana propagation, as explained in Appendix~\ref{app:pauli_to_majorana}. \\

\paragraph*{Coefficient truncation.}
We start by analyzing coefficient truncation schemes whereby after each update we discard Pauli terms whose amplitudes are below some cutoff. 
For the analysis, however, it is convenient to work with a closely related proxy rule that is simpler to treat analytically. 
Namely, we consider a \emph{small-angle} truncation scheme which exploits the fact that for $\tau\ll 1$ the branching amplitudes produced by imaginary-time propagation are suppressed in powers of $\sinh(\tau) \approx \tau \ll 1$. Therefore we can discard any Pauli path that accumulates more than $k$ non-identity updates (equivalently, more than $k$ factors of $\sinh(\tau)$) as $\sinh(\tau)^k \approx \tau^k \ll 1$. 
This results in the following theorem.

\begin{theorem}[Small-angle truncation error]
\label{thm:qDRIFT_truncation-informal}
Consider the simulation of the thermal state starting from the identity $\1$. The algorithm performs $L$ steps of the map $\mathcal{E}_t(\cdot) = e^{-\tau P_t/2}(\cdot) e^{-\tau P_t/2}$, with step angle $\tau = \frac{\beta \Lambda}{L}$, where $\Lambda = \sum_j |h_j|$ is the sum of Hamiltonian coefficients. Let $\tilde\rho$ be the approximate state obtained by truncating any Pauli path that accumulates more than $k$ non-identity updates (i.e., discarding paths with more than $k$ factors of $\sinh\tau$). 
Then the approximate thermal state $\tilde \rho$ satisfies
\begin{align}
    \norm{\rho - \tilde \rho}_1 \in \mathcal{O}\left(\, e^{\beta\Lambda/2} \left(\frac{e\beta\Lambda}{2k}\right)^k \right) \, .
\end{align}
As a consequence, for any observable $O$ one has
    \begin{align}
        \abs{\Tr[O(\rho-\tilde\rho)]} \in \mathcal{O}\left( \|O\|_\infty \, e^{\beta\Lambda/2} \left(\frac{e\beta\Lambda}{2k}\right)^k \right) \, .
    \end{align}
\end{theorem}

This is proven in Appendix~\ref{app:small_angle}. Thus, for temperatures such that $\beta$ scales inversely with the $1$-norm of the Hamiltonian (i.e., $\beta\Lambda \in \mathcal{O}(1)$), the truncation error is suppressed super-exponentially in $k$:
\begin{equation}
    \abs{\Tr[O(\rho-\tilde\rho)]}
    \in
    \mathcal{O}\!\left(\|O\|_\infty\, 
    \left(\frac{c}{k}\right)^k\right)
    \in 
    \exp\!\bigl(-\Omega(k\log k)\bigr)
\end{equation}
for $k\gg c$ for some constant $c$. Equivalently, to achieve a target accuracy $\epsilon$, it suffices to choose
\begin{equation}
    k \;=\; \Theta\!\left(\frac{\log\!\left(\frac{\|O\|_\infty }{\epsilon}\right)}{\log\log\!\left(\frac{\|O\|_\infty }{\epsilon}\right)}\right),
\end{equation}
which grows only slightly faster than $\log(1/\epsilon)$. 

Since the number of retained terms (and hence the runtime of propagation) scales polynomially with the truncation level $k$ for fixed system size and locality structure~\cite{lerch2024efficient, angrisani2024classically}, this implies that for $\beta\in\mathcal{O}(1/\Lambda)$ the small-angle imaginary-time simulation error decays super-exponentially in $k$,
\begin{equation}
    \epsilon \;=\; \|O\|_\infty \,\exp\!\bigl(-\Omega(k\log k)\bigr).
\end{equation}
in time polynomial in $k$ (and therefore is polylogarithmic in $1/\epsilon$, up to the $\log\log$ correction). Furthermore, as proven in Lemma~\ref{lemma:scaling_bound_small_angles} in Appendix~\ref{app:small_angle}, for slightly lower temperatures such that $\beta\Lambda\in\Theta(\log n)$ we can obtain an error that is polynomially suppressed in $n$ for $k \in \mathcal{O}(\log(n))$ which corresponds to super-polynomial resources. We stress that as Theorem~\ref{thm:qDRIFT_truncation-informal} also bounds the error of the 1-norm of the state and so the simulated state can be used for sampling-based applications~\cite{Minervini2026Strong}.

\medskip
\paragraph*{Weight truncation.}
We now provide efficiency guarantees for imaginary-time Pauli propagation under weight truncation, where we discard any Pauli string whose weight exceeds a threshold $k$. The same logic can also be applied to Majorana length truncation. 

The intuition underlying this truncation scheme begins with the observation that only Pauli strings that overlap the support of the observable $O$ can contribute to $\Tr(O\rho)$. Thus, if $O$ is low weight, high-weight strings matter only insofar as they subsequently \emph{backflow} to low weight. Any path that realizes such a large net decrease in weight must branch many times, and therefore accumulates multiple factors of $\sinh(\tau)$. Since $\sinh(\tau)\approx \tau \ll 1$ for small $\tau$, these contributions acquire very small coefficients and induce only a small error. This effect is further strengthened by a typical drift toward higher weight: for a Pauli string of weight $w\lesssim n/2$, applying a local imaginary-time gate is more likely to increase weight than to decrease it. Consequently, the overall contribution of truncated high-weight terms is additionally suppressed by the small probability of backflow from weight $k$ down to the (low) weight scale relevant to $O$.

Combining these effects yields a bound analogous to coefficient truncation, but with an additional suppression factor governed by the backflow probability. Let $\pbf(w)$ denote the probability (over the random choice of gate e.g.\ under qDRIFT) that applying one imaginary-time gate to a Pauli string of weight $w$ produces a branch whose resulting Pauli string has \emph{strictly smaller} weight. As shown in Appendix~\ref{app:weight_trunc}, this quantity for a range of models typically depends on the current weight $w$ of the propagated string but does not vary significantly between different strings of the same weight. We further show that $\pbf(w)$ typically increases with $w$ since higher-weight strings have more opportunities for cancellations that reduce support. 
When deriving a worst-case bound for weight truncation, we upper bound all backflow events originating from the truncated sector $w>k$ (as any path that starts backflowing from higher will also have to backflow from $k$ down to the low weight region). Therefore, the maximum backflow probability is 
\begin{equation}
    q_{\max}  \;:=\;\sup_{w\leq k}\,\pbf(w) \;:=\; \pbf(k) \, .
\end{equation}
This results in the following bound.

\begin{theorem}[Weight truncation Backflow Error, Informal]
\label{thm:qDRIFT_backflow-informal} Consider the qDRIFT simulation of a thermal state starting from the identity $\1$ and let $O$ be a constant weight observable. The algorithm applies $L$ steps of size $\tau = \frac{\beta \Lambda}{L}$, where $\Lambda = \sum |h_i|$. Let $\tilde\rho$ be the state obtained by truncating any Pauli string with weight exceeding $k$. 
For sufficiently large $k \in O(1)$, the error in the expectation value is bounded by:
    \begin{align}
        \abs{\Tr[O(\rho-\tilde\rho)]} \in \mathcal{O}\left( \|o\|_1 \, e^{\beta\Lambda/2} \left(\frac{e\beta \Lambda\,\pbf(k) }{2k}\right)^{k} \right) \, ,
    \end{align}
where $\|o\|_1$ is the 1-norm of its Pauli coefficients of $O$.
For approximately uniform all-to-all Hamiltonians the backflow probability is quadratically suppressed,
\begin{equation}
    \pbf(k) \in \mathcal{O}(k^2/n^2) \, ,
\end{equation}
where as for 1D Hamiltonians it is linearly suppressed,
\begin{equation}
    \pbf(k) \in \mathcal{O}(k/n) \, . 
\end{equation}
\end{theorem}
This theorem is proven in Appendix~\ref{app:weight_trunc} for a general $\pbf(k)$, and in Appendix~\ref{app:pbf} we compute $\pbf(k)$ for different Hamiltonians. 

Since Theorem~\ref{thm:qDRIFT_backflow-informal} has the same functional form as Theorem~\ref{thm:qDRIFT_truncation-informal}, we again obtain super-exponential suppression of the truncation error in the high-temperature regime $\beta\Lambda=\mathcal{O}(1)$. Moreover, the additional factor $\pbf(k)$ in the exponent further suppresses contributions from truncated high-weight strings and therefore permits meaningful guarantees at moderately larger $\beta$. Concretely, when $\beta\Lambda=\mathcal{O}(\log n)$, the error becomes exponentially small in the truncation order. In particular, for sufficiently large (but still $\mathcal{O}(1)$) truncation threshold $k$, we obtain the scaling
\begin{align}
    \epsilon \in \left(\frac{\log n}{n}\right)^{\Theta(k)}\,.
\end{align}
For details, see the proof of Lemma~\ref{lemma:scaling_bound_wbp} in Appendix~\ref{app:weight_trunc}.

\begin{figure*}
    \centering
    \includegraphics[width=0.9\linewidth]{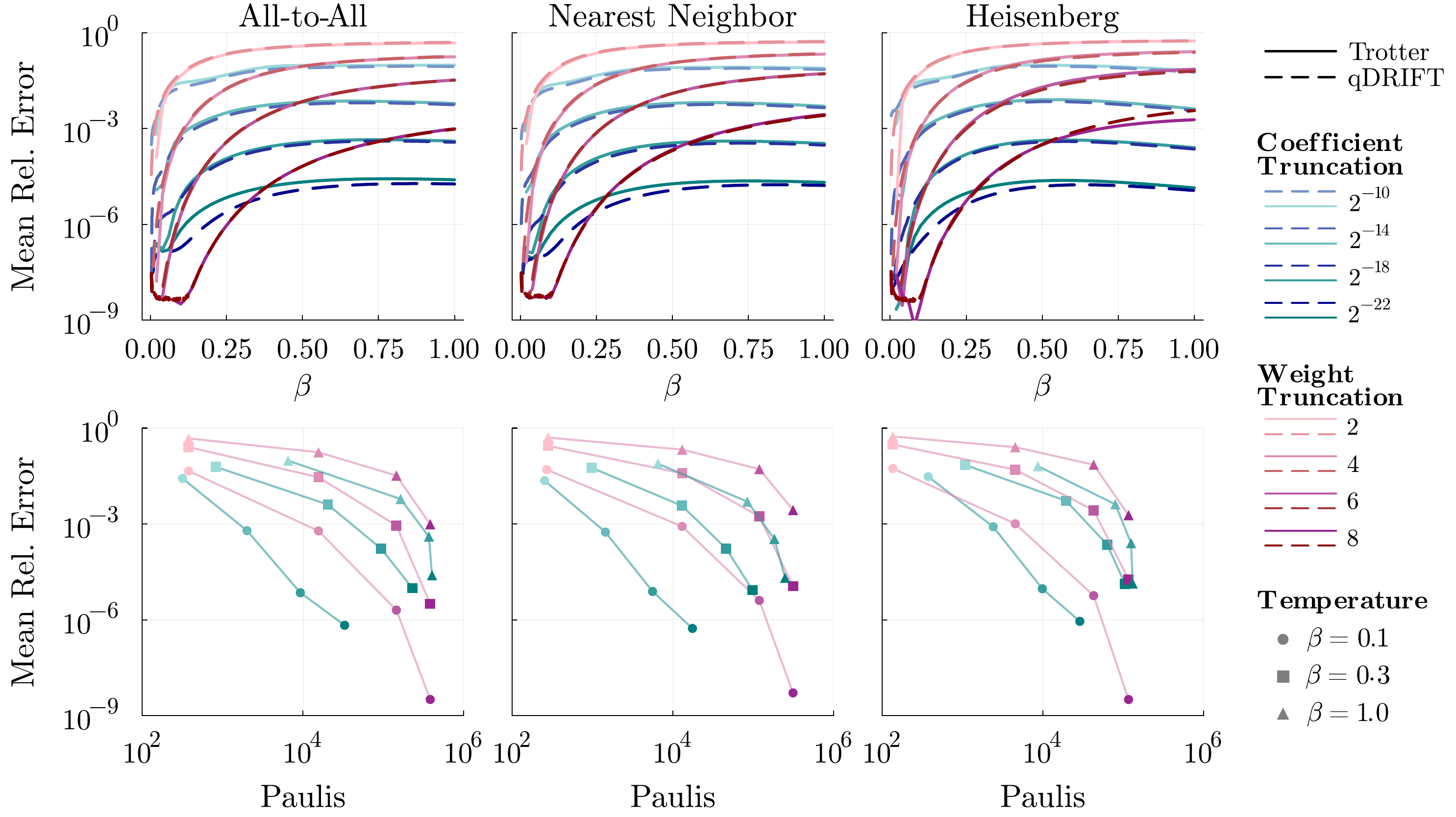}
    \caption{\textbf{Comparison between 1st-order Trotter and qDRIFT.} We evolve using truncated trotterization (full) and qDRIFT (dashed) for different models and compare to the untruncated evolutions using the mean relative error of the energy as a metric. This error is plotted as a function of $\beta$ (top) and as a function of the number of Paulis it utilizes (bottom) for a 10 qubit system. The overlap of the curves between the trotter evolution and the qDRIFT suggests that weight truncation is equally justified for both types of evolution. The bottom plot gives a comparison of the effectiveness of each truncation scheme where in general the coefficient truncation leads to a better error at a given number of Paulis. }
    \label{fig:qdrift}
\end{figure*}

In Appendix~\ref{app:pbf}, we argue that most physically relevant Hamiltonians can be viewed as close-to-uniform and therefore exhibit between linearly to quadratically suppressed backflow depending on their topology. 

We emphasize that the bound above is stated for qDRIFT because it provides a natural probabilistic language for defining and analyzing backflow. Nevertheless, the same intuition carries over to deterministic product formulas: in a standard Trotter scheme, a significant proportion 
of high-weight contributions continue to grow in weight rather than backflow to low weight, so the net influence of truncated high-weight terms on local observables is suppressed. We can capture this via a direct combinatorial (frequentist) analysis of Trotter paths, the broad intuition is the same as in qDRIFT but the analysis is technically more awkward.
In Appendix\ \ref{app:wt-trunc-trotter}, we obtain the following bound.

\begin{theorem}[Weight-truncation error for 1st-order Trotter imaginary-time evolution (informal)]
\label{thm:trotter_weighttrunc_informal}
Let $H=\sum_{a=1}^M \lambda_a Q_a$ be a Pauli Hamiltonian on $n$ qubits with
$|\lambda_a|\le 1$ and $|\supp(Q_a)|\le w$ for all $a$.
Assume $H$ has bounded degree $\ell$, i.e.\ each qubit participates in at most $\ell$ terms.
Fix an inverse temperature $\beta>0$ and a step size $\delta>0$, and let $L:=\beta/\delta\in\mathbb N$.
Let $\rho$ denote the (normalized) state obtained by applying $L$ steps of the 1st-order Trotter
imaginary-time evolution, and let $\tilde\rho$ be the classical approximation obtained by truncating
all propagated Pauli operators to weight $<k$ after any imaginary time operation.

Then there exist absolute constants $c_1,c_2>0$ such that, for every observable $O$,
\begin{equation}\label{eq:bound-trotter}
\bigl|\Tr\!\bigl[O(\rho-\tilde\rho)\bigr]\bigr|
\;\le\;
\|o\|_1\;
\exp\!\bigl(c_1\,\beta M\bigr)\;
\Bigl(c_2\,\beta\,\ell\,w\Bigr)^{ k/w }.
\end{equation}
where $\|o\|_1$ is again the $1$-norm of the Pauli coefficients of $O$. 
\end{theorem}

\noindent
We note that the error obtained here aligns with the scalings obtained in Theorem~\ref{thm:qDRIFT_backflow-informal}.
Consider the high-temperature regime $\beta M= \Theta(1)$ and $w=\mathcal O(1)$.
Then \eqref{eq:bound-trotter} implies
\begin{align}
\bigl|\Tr[O(\rho-\tilde\rho)]\bigr|
\;\le\; &
\|o\|_1\;
\bigl(\beta\,\ell\bigr)^{\Omega(k)}\\=
&\|o\|_1\;
\Bigl(\frac{\ell}{M}\Bigr)^{\Omega(k)}.
\end{align}
Under the natural assumption that the degree is sublinear in the number of terms, e.g.\ $\ell\le M^c$
for some constant $c<1$, we obtain a polynomially small error in $M$ with exponent linear in $k$:
\[
\bigl|\Tr[O(\rho-\tilde\rho)]\bigr|
\;\le\;
\|o\|_1\;
M^{-\Omega(k)}.
\]
Similarly to our bound in Theorem~\ref{thm:qDRIFT_backflow-informal}, Theorem~\ref{thm:trotter_weighttrunc_informal} yields meaningful guarantees also  for the case where $\beta M \in \Theta(\log(M))$. 
In this regime, assuming also $w\in \mathcal{O}(1)$ and that the degree is sublinear in the number of terms, e.g.\ $\ell\le M^c$ for some constant $c<1$, one has
\begin{align}\label{eq:bound-trotter-log}
\bigl|\Tr\!\bigl[O(\rho-\tilde\rho)\bigr]\bigr|
\;\in\;&
\|o\|_1\;
M^{\mathcal{O}(1)}\;
\left(\frac{\log(M)}{M^{1-c}}\right)^{\Omega(k)}
\\ \in\;&
\|o\|_1\; M^{\mathcal{O}(1)-\Omega(k)}.
\end{align}

Exactly which of the bounds in Theorem~\ref{thm:trotter_weighttrunc_informal} and Theorem~\ref{thm:qDRIFT_backflow-informal} is tighter depends on the relative magnitude of the degree $l$ of the Hamiltonian and the truncation $k$ considered which will vary on a case-by-case basis. However, we suspect these differences are not fundamental but rather are an artifact of our proof techniques. Indeed this is further supported by our numerics below. \\

\begin{figure*}
    \centering
    \includegraphics[width=0.99\linewidth]{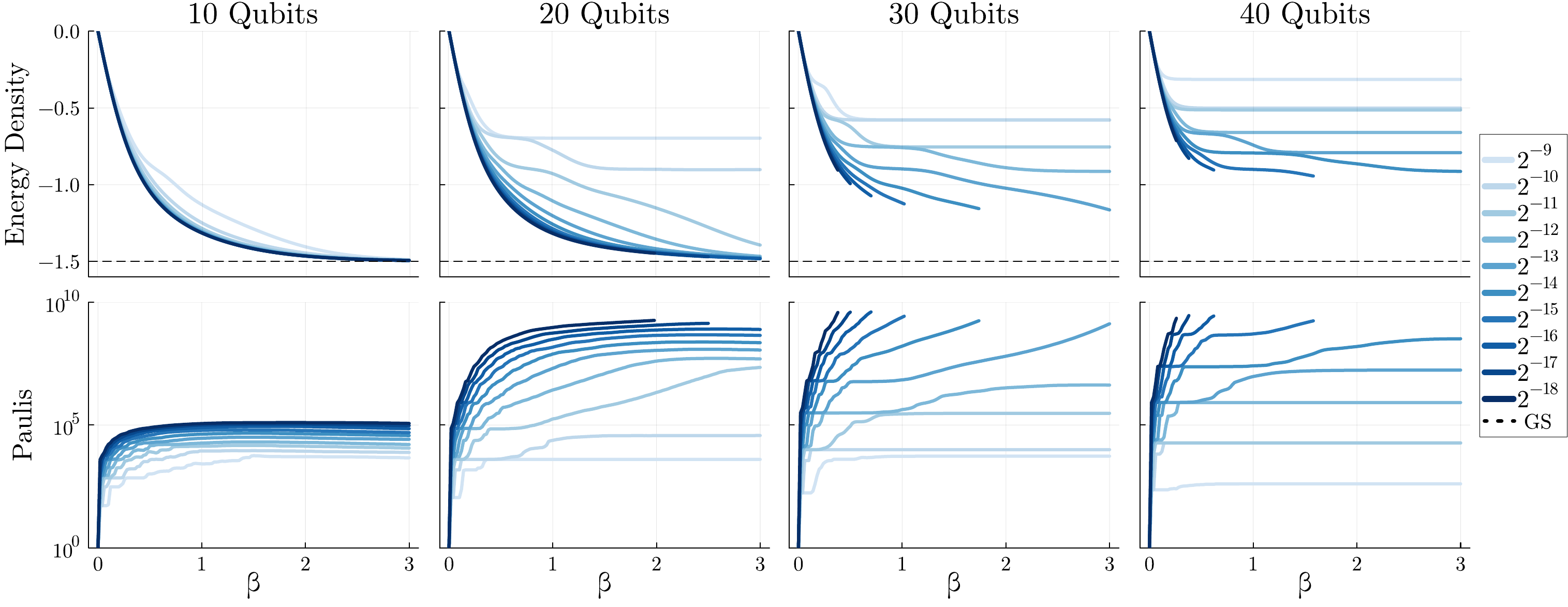}
    \caption{\textbf{Imaginary time evolution under the J1-J2 Hamiltonian.}  Numerical results for the 1D $J_1-J_2$ Heisenberg model across system sizes ranging from 10 to 40 qubits. The top row shows the estimated energy density $\langle H \rangle / n$ converging toward the ground state (dashed line) as the inverse temperature $\beta$ increases. The bottom row tracks the number of Pauli strings generated, showing the exponential growth in complexity as the temperature drops. Different shades of blue represent different coefficient truncation thresholds (from $2^{-9}$ to $2^{-18}$); the point where these lines diverge indicates the limit where our approximation is no longer accurate.}
    \label{fig:j1j2-numerics}
\end{figure*}

\paragraph*{Numerically probing theoretical predictions.}

In Fig.~\ref{fig:qdrift}, we numerically illustrate the effects of coefficient and weight truncation for imaginary time evolution via qDRIFT or a 1st-order Trotter decomposition. This is to test whether  Theorem~\ref{thm:qDRIFT_backflow-informal} and Theorem~\ref{thm:trotter_weighttrunc_informal} point to more general effects that transcend our proof techniques. We consider three types of Hamiltonians, two random and one physical one, all composed of 30 weight-2 Paulis: Random all-to-all interactions with coefficients $\in\{-1,1\}$, random nearest-neighbor interactions with coefficients $\in\{-1,1\}$, and a 1D Heisenberg Hamiltonian with $+1$ coefficients. The simulations were done on 10-qubit systems for exact verification. In the Trotter case, we use $L = 50$ steps with $\tau= 0.02$, and for qDRIFT $L_{qD} = \Lambda \beta/\tau = \Lambda L = 1500$ with $\Lambda = 30$ to match the theoretical time discretization error.
For each model, we consider 200 samples and compute the mean relative error of the energy between the truncated and exact simulations.

There appears to be hardly any difference between how Trotter or qDRIFT respond to truncations, indicating that both our theoretical guarantees for weight truncation have similar predictive power. When comparing coefficient and weight truncation, we observe that, for a given temperature, coefficient truncation obtains lower error per propagated Pauli operator, which is a proxy for runtime and memory consumption. This is in large part because coefficient truncation integrates one partial effect of weight truncation. Namely, high-weight terms have accumulated more sinh (or sin) coefficients to become high weight and thus tend to have smaller coefficients and thus as naturally truncated anyway via coefficient truncation. Below we outline further practical considerations and numerical results.

\section{Numerical implementation}

\begin{figure*}
    \centering
    \includegraphics[width=0.9\linewidth]{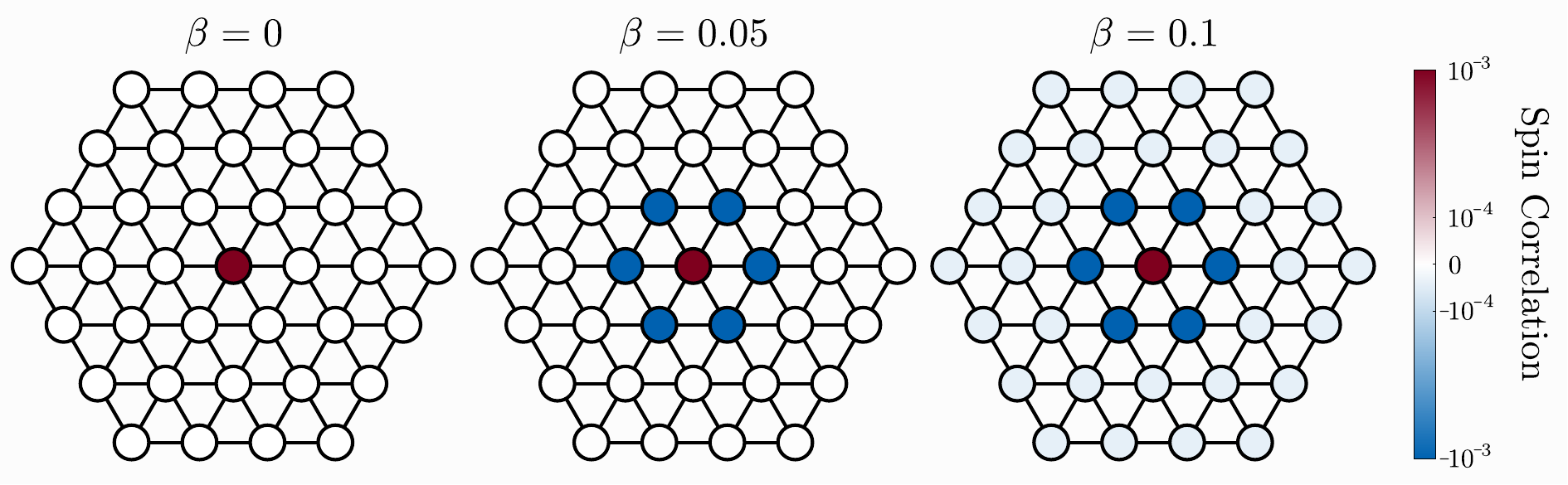}
    \caption{\textbf{Spin correlation in the Fermi-Hubbard model on a triangular lattice.} Snapshots of spin-spin correlations ($C_{ZZ}$) emerging in a Fermi-Hubbard model on a 37 site triangular lattice (corresponding to 74 Majorana modes). As we lower the temperature from $\beta=0$ to $\beta=0.1$, we can observe magnetic order and correlations beginning to form around the central site (red). While reaching the deep frustration patterns of lower temperatures ($\beta \sim 2$ to $\beta \sim 3$) remains challenging for current propagation techniques, this demonstrates the ability to compute static correlations for complex geometries directly from the propagated state.}
    \label{fig:fh-magnetization}
\end{figure*}

\paragraph*{Practitioners notes.}
In this section, we discuss practical considerations for performant numerical implementations, both ones that are specific to imaginary time evolution, as well as ones that apply more broadly to propagation methods.

For numerical stability and to facilitate coefficient truncation, we (i) continually normalize our representation of the state by ${1}/{\alpha_\mathbb{I}}$ after every gate, and (ii) never divide by the factor $2^n$ in Eq.~\eqref{eq:rho_beta}. This is a variant on the necessary normalization by $\Tr[\rho]$ that keeps our coefficients at unit scale and less dependent on the number of qubits (at least at high temperatures). Our coefficient truncation threshold can then be seen as a relative truncation threshold to one of the largest terms, the identity operator. Neglecting the $2^n$ factor has a flavor of working in the \textit{normalized} Pauli basis, and the factor is automatically, implicitly regained by calculating expectation values by proceeding as if $\Tr[\mathbb{I}]=1$. Both practices together provide a natural and convenient rescaling.

Then we would like to draw attention to the choice of truncation method. Truncated numerical simulations often perform substantially better than their theoretical guarantees. At the same time, some truncation strategies are more amenable to theoretical analysis than others. Putting both observations together, we realize that the best truncation may not be the one that has the most (or any) attractive bound. For imaginary time evolution, similarly to previous experiences with real time evolution, we have found that \emph{coefficient truncation is most effective in general}. We stress that coefficient truncation will be strictly better than the `small angle' truncation scheme analyzed above because we truncate on the exact values of the coefficients which will depend on the Hamiltonian parameters and can be larger due to merging different paths that produce the same Pauli.  

Pauli weight and Majorana length truncation do have their distinct effects for estimating local expectation values~\cite{angrisani2024classically,miller2025simulation} that go beyond what coefficient truncation alone can enable. However, we find that in combination with coefficient truncation they often add little and so it is preferable to focus on just coefficient truncation. This is particularly relevant when aiming to converge a simulation, i.e., to gain confidence in the result, which is less cumbersome with just one truncation parameter. Hence in our large scale numerics below we use only coefficient truncation. 

Finally, we would like to highlight our numerical implementation. Until recently, PauliPropagation.jl natively supported only single-threaded dictionary-based propagation. While multi-threaded CPU and even GPU implementations were possible, one would likely run out of memory trying to simulate real-time dynamics rather than run out of time, even with a single CPU thread. Imaginary time evolution, on the other hand, appears to be noticeably slower, owing to its accelerated branching upon commutation akin to the Heisenberg evolution of global observables in the real-time case. Truncations can keep operator numbers in check, but the time per gate remains elevated. We are pleased to report that our large-scale results were generated with array-based propagation that was at least an order of magnitude faster by use of up to 72 CPU threads with off-the-shelf code publicly available via PauliPropagation.jl. GPU acceleration is already possible, but faces continued limitations due to the lower memory of everything but the highest-end GPU hardware compared to conventional HPC CPU nodes. Our work showcases that innovations in the applications can motivate innovation in the algorithms, which we believe will continue in the future. \\

\paragraph*{The J1-J2 Heisenberg Model.}
As a first practical demonstration of our algorithm, we study the one-dimensional Heisenberg model with both nearest-neighbor ($J_1$) and next-nearest-neighbor ($J_2$) interactions. The Hamiltonian is defined as
\begin{align}
    H &= J_1 \sum_{i} \left(X_iX_{i+1} + Y_iY_{i+1} + Z_iZ_{i+1} \right)\nonumber\\
    &+J_2 \sum_{i} \left(X_iX_{i+2} + Y_iY_{i+2} +  Z_iZ_{i+2}\right),
    \label{eq:j1j2_hamiltonian}
\end{align}

\noindent where $X_i, Y_i, Z_i$ are Pauli operators at site $i$.
The terms with the $J_1$ coefficient represent the standard nearest-neighbor Heisenberg exchange. Assuming an antiferromagnetic coupling ($J_1 > 0$), this term favors an antiparallel alignment of adjacent spins. 
The $J_2$ terms introduce a next-nearest-neighbor interaction. When $J_2 > 0$, this term acts as a source of geometrical frustration, as it competes with the ordering tendency of the $J_1$ term. Here, we use $J_1 = 1.0$ and $J_2=0.5$, with a Trotter step of $\tau=0.02$

Fig.~\ref{fig:j1j2-numerics} depicts our results for preparing states at increasing inverse temperature $\beta$ for $10$, $20$, $30$, and $40$ qubits using $\tau=0.02$ in a first-order Trotter approximation, and for decreasing coefficient truncation thresholds (relative to the identity operator) ranging from $2^{-9}$ to $2^{-18}$. We report the energy density, i.e., $\langle H\rangle/n$, as well as the number of Pauli strings that are being generated. The black dashed line at energy density $-1.5$ denotes the energy density of the ground state in this model, which can be computed to numerical precision using tensor network techniques~\cite{white1992density}. It is evident that our approach is increasingly inefficient in reaching low-temperature states, yet high-temperature states remain trackable at scale. The largest $\beta$ for which the simulations can be trusted can be estimated by the convergence of overlapping lines at lower and lower truncations. \\

\paragraph*{Fermi-Hubbard model on a triangular lattice.}
We further consider the Fermi-Hubbard model of interacting Fermions on a triangular lattice with radial hexagonal geometry. The system is described by the Hamiltonian
\begin{align}
    H = &-t \sum_{\langle i,j \rangle, \sigma} (c_{i\sigma}^\dagger c_{j\sigma} + c_{j\sigma}^\dagger c_{i\sigma} ) \nonumber\\
    &+ U \sum_i n_{i\uparrow} n_{i\downarrow}
    -\mu \sum_i n_i,\label{eq:FH-hamiltonian}
\end{align}
where $c_{i\sigma}^\dagger$ ($c_{i\sigma}$) creates (annihilates) a Fermion with spin $\sigma \in \{\uparrow, \downarrow\}$ at site $i$, $n_{i\sigma} = c_{i\sigma}^\dagger c_{i\sigma}$ is the number operator, $\mu$ the chemical potential and $n_i =  n_{i\uparrow} + n_{i\downarrow}$. 
We explore the strong coupling regime with interaction strengths $t = 1, U =8$, similar to other realizations in, for example, Refs.~\cite{xu2023frustration,sinha2022finite}. We tune the chemical potential to $\mu=U/2$, which stabilizes the average particle number to approximately equal the number of sites, i.e., to half-filling. Being mindful of Trotter error and the reduced simulation horizon of this more challenging system, we employ a smaller Trotter step of $\tau=0.01$ and a randomized order in the evolution under the hopping terms. 

As the temperature sinks, magnetic order and correlations start to form. We characterize it by the spin-spin correlation function $C_{ZZ_i}$ between the center site $r$ and site $i$. It reads
\begin{equation}
    C_{ZZ_i} = \langle Z_r Z_{i} \rangle - \langle Z_r \rangle \langle Z_{i} \rangle,
\end{equation}
where $Z_i = n_{i\uparrow} - n_{i\downarrow}$.

Our results can be seen in Fig.~\ref{fig:fh-magnetization}, where in a system with 37 sites (corresponding to 74 Majorana modes) we observe correlations build up as we lower the temperature to $\beta=0.1$. Ref.~\cite{xu2023frustration,sinha2022finite} shows by use of a neutral atom analog simulator that thermal states at lower temperature ($\beta\sim2$ to $\beta\sim3$) exhibit frustration and staggered correlation patterns, temperatures which we can unfortunately not access with current propagation techniques. Yet, it is noteworthy that we can start tackling such systems with off-the-shelf code, and explore both physical phenomena and the potential use of propagation methods.

\section{Discussion}

In this work, we introduce a framework for simulating thermal states via Pauli and Majorana propagation. It utilizes imaginary time evolution applied to maximally mixed states in the Schrödinger picture to continually lower the temperature. At the end of the simulation, one receives the quantum states ``in plain text'', and can readily read off properties of interest. Our high-performance numerical implementation effectively leverages multi-CPU computing resources and can quickly generate billions of Pauli or Majorana strings. 

We further provide analytical guarantees for both coefficient and weight truncation strategies which hint toward propagation algorithms being highly capable of simulating high-temperature states in arbitrary system geometries. 
At the same time, both our analysis and numerical results indicate that lower temperatures where $\beta \sim 1$ can be very challenging to accurately simulate for physically motivated problems. 

Going forward, we would like to highlight three research directions. One is the 
question of when propagation methods are the algorithm of choice for simulating quantum states. Promising cases include repeated evaluations of high-temperature states with complicated connectivities. This lends itself, for example, to quantum Boltzmann machine training~\cite{Amin2018Quantum}. Another direction is that of ground state exploration. While sheer numbers indicate that zero-temperature states are not efficiently representable in Pauli or Majorana basis due to being exponentially dense, truncated propagation may be useful for estimating or extrapolating ground state energies (see, e.g., Ref.~\cite{shrikhande2025rapid}). Finally, we would like to highlight propagation algorithms natural capabilities in hybridizing with quantum hardware. We can envision that unitary circuits for approximate thermal state preparation could be learned or compressed classically~\cite{danna2025circuit} for further processing on quantum devices once they are mature enough.

\medskip

\paragraph*{Code Availability.}
The Pauli propagation results in this work have been obtained with the open-source library \hyperlink{https://github.com/MSRudolph/PauliPropagation.jl}{PauliPropagation.jl}, which now supports imaginary Pauli rotations. The Majorana propagation results in this work have been obtained with the open-source library \hyperlink{https://github.com/SparqleSim/MajoranaPropagation.jl}{MajoranaPropagation.jl}, which is not yet registered and will soon support imaginary Majorana rotations. \\

\paragraph*{Acknowledgements.}
The authors would like to thank Matteo D'Anna for his work on MajoranaPropagation.jl.  The authors would like to thank Yanting Teng for insightful discussions.
MSR acknowledges funding from the 2024 Google PhD Fellowship and the Swiss National Science Foundation [grant number 200021-219329]
AA and ZH acknowledge support from the Sandoz Family Foundation-Monique de Meuron program for Academic Promotion. RP acknowledges the support of the SNF Quantum Flagship Replacement Scheme (Grant No. 215933). \\

\paragraph*{Added note.} In the final stages of preparing this manuscript, the authors became aware of independent work demonstrating imaginary time evolution for Pauli propagation~\cite{gomez2026pauli}.

\bibliography{references,quantum}

\clearpage 

\appendix
\onecolumngrid

\section{Notation}
We denote the $n$-qubit Pauli basis by $\mathcal{P}_n \coloneqq \{\mathbb{I}, X, Y, Z\}^{\otimes n}$. Any operator $O$ admits a Pauli expansion
\[
O=\sum_{P\in\mathcal P_n} c_P\,P.
\]
We define the \emph{Pauli-$1$ norm} of $O$ as
\[
\|O\|_{\mathrm{Pauli},1}\;:=\;\sum_{P\in\mathcal P_n} |c_P|.
\]
When no confusion can arise, we use the shorthand $\|o\|_1:=\|O\|_{\mathrm{Pauli},1}$.

\section{Preliminaries}
In this section, we present basic lemmas that are required for the subsequent proofs.

\begin{lemma}\label{lemma:small_sinhcosh}
    Let $r(\gamma)$ be the number of times the $\sinh(\tau)$ branch is taken along a path $\gamma$.
    Then for all $\tau\in(0,1]$,
    \begin{equation}
        \cosh(\tau)^{L-r(\gamma)}\,\sinh(\tau)^{r(\gamma)}
        \le e^{L\tau^2/2}\,\tau^{r(\gamma)}.
    \end{equation} 
\end{lemma}
\begin{proof}
For $\tau\ge 0$ we have $\sinh(\tau)\le \tau\cosh(\tau)$. Hence 
\[
\cosh(\tau)^{L-r(\gamma)}\,\sinh(\tau)^{r(\gamma)}
\le \cosh(\tau)^{L-r(\gamma)} \, (\tau \cosh(\tau))^{r(\gamma)}
= \tau^{r(\gamma)} \cosh(\tau)^L.
\]
Using $\cosh(\tau)\le e^{\tau^2/2}$, we get 
\[
\tau^{r(\gamma)} \cosh(\tau)^L \le \tau^{r(\gamma)} e^{L\tau^2/2}.
\]
\end{proof}

\begin{lemma}\label{lemma:lowerbound_z}
    Consider a matrix $\nstate = \prod_{j = 1}^J e^{-\tau P_j}\1\prod_{j = J}^1 e^{-\tau P_j}$. Then $\Tr[\nstate] \ge 2^n$.
\end{lemma}
\begin{proof}
    First, observe that $\det(e^{-\tau P_j})= e^{-\tau \Tr[P_j]} = e^0 = 1$ for any Pauli operator $P_j$. Consequently, using the multiplicativity of the determinant, we have $\det(\nstate) = 1$.
    Then since $\nstate$ is of the form $A^\dagger A$ (with $A = \prod_{j = J}^1 e^{-\tau P_j}$), it is positive semidefinite ($\nstate \succeq 0$). Let $d=2^n$ be the dimension of the Hilbert space. Applying the Arithmetic Mean-Geometric Mean inequality to the eigenvalues of $\nstate$, we obtain:
    \begin{equation}
        \frac{\Tr[\nstate]}{d} \geq \det(\nstate)^{1/d} = 1.
    \end{equation}
    Multiplying by $d$ yields $\Tr[\nstate] \ge 2^n$.
\end{proof}

\begin{definition}[Max Divergence]\label{def:max_divergence}
Let $p$ and $u$ be two probability distributions on a set $E$ of size $M$. The (forward) max-divergence of $p$ from $u$ is defined as:
\begin{align}
    D_\infty(p\|u)
    :=\log\!\left(\max_{e\in E}\frac{p(e)}{u(e)}\right)
    =\log\!\left(M\cdot \max_{e\in E} p(e)\right).
\end{align}  
\end{definition}

\begin{lemma}\label{lemma:entropy_bound}
    Let $u$ denote the uniform distribution on a finite set $E$, i.e., $u(e)=1/M$. Let $p$ be any distribution on $E$. If $D_\infty(p\|u)\le \alpha$, then for every event $A\subseteq E$,
    \begin{equation}\label{eq:Dinfty-event}
        p(A)\ \le\ e^\alpha\,u(A).
    \end{equation}
\end{lemma}
\begin{proof}
    The assumption $D_\infty(p\|u)\le\alpha$ implies pointwise domination:
    \begin{align}
        p(e) \le e^\alpha u(e) \qquad \forall e\in E.
    \end{align}
    Summing this inequality over all elements $e\in A$ yields the result:
    \begin{align}
        p(A)=\sum_{e\in A}p(e)\le e^\alpha\sum_{e\in A}u(e)=e^\alpha u(A).
    \end{align}
\end{proof}

\begin{proposition}\label{prop:random_ham}
    Let $E$ be a set of size $M$. Let weights $(W_e)_{e\in E}$ be drawn i.i.d.\ from $\mathrm{Unif}(0,1)$, and define $p(e) = W_e / \sum_{f} W_f$. Let $u$ be the uniform distribution over $E$.
    Then, with probability at least $1-e^{-M/8}$, $D_\infty(p\|u) \le \log 4$.
\end{proposition}

\begin{proof}
    Let $S:=\sum_{f\in E} W_f$. Since $\mathbb{E}[W_f]=1/2$, $\mathbb{E}[S]=M/2$. By Hoeffding's inequality:
    \begin{align}
        \Pr\left[S \le \frac{M}{4}\right] 
        \le \exp\left( -\frac{2 (M/4)^2}{M} \right)
        = e^{-M/8}.
    \end{align}
    Conditioned on $S \ge M/4$, for any $e \in E$:
    \begin{align}
        \frac{p(e)}{u(e)} = \frac{W_e/S}{1/M} = \frac{M W_e}{S} \le \frac{M \cdot 1}{M/4} = 4.
    \end{align}
    Taking the logarithm gives $D_\infty(p\|u) \le \log 4$.
\end{proof}

\section{General truncation error}
Here we present a general error bound for a Pauli Propagation scheme where the simulation is restricted to a specific subset of Pauli strings. This will be the starting point to derive both our small angle and weight truncation guarantees for the case of observables. 

\begin{lemma}\label{lemma:general_error}
Let $O$ be a $\owt$-local observable, $O=\sum_{\wt(Q)\le \owt} o_Q Q$, and let $\nstate_L$ be the unnormalized state obtained after applying $L$ steps of an imaginary time evolution $e^{-\tau P_t}(\cdot) e^{-\tau P_t}$, where $P_t$ is sampled randomly at each step $t$.
Let $\mathcal{A} \subseteq \mathcal{P}_n$ be an arbitrary set of ``allowed'' Pauli strings (containing the identity). Let $\widetilde\nstate_L$ be the matrix produced by truncating any Pauli string $Q \notin \mathcal{A}$ at any step. Then, the expected error in the observable expectation value is bounded by
\begin{align}
    \mathbb E\big[ |\Tr\!\big(O(\nstate_L-\widetilde\nstate_L))|\big]
    \leq 2^n\|o\|_1 \, \mathbb E\Bigg[ \sum_{\substack{\gamma\notin\Gamma_{\mathcal{A}} \\ \wt(Q_L(\gamma))\le \owt}}  \cosh(\tau)^{L-r(\gamma)}\,\sinh(\tau)^{r(\gamma)} \Bigg] \, ,
\end{align}
where $\Gamma_{\mathcal{A}}$ is the set of paths that remain strictly within the set $\mathcal{A}$ at all steps, and $r(\gamma)$ is the number of non-identity updates in path $\gamma$.
\end{lemma}

\begin{proof}
A \emph{Pauli path} $\gamma=(Q_0,Q_1,\dots,Q_L)\in \mathcal P_n^{L+1}$ with $Q_0=\mathbb{I}$ is a sequence where each $Q_t$ is obtained from $Q_{t-1}$ by applying the local update rule in Eq.~\eqref{eq:branching-rule} for the sampled term $P_t$. For such a path $\gamma$, we denote the Pauli string at step $t$ by $Q_t(\gamma)$, and the product of the accumulated $\cosh(\tau)$ and $\sinh(\tau)$ factors by $\coef(\gamma)$.
We define the set of \emph{admissible} paths for the truncated recursion as those that never leave the allowed set $\mathcal{A}$:
\begin{align}
    \Gamma_{\mathcal{A}}:=\{\gamma \mid \forall t \le L, \, Q_t(\gamma) \in \mathcal{A}\}.
\end{align}
The truncated evolution $\widetilde\nstate_L$ has the same path expansion as the full evolution but is restricted to $\Gamma_{\mathcal{A}}$:
\begin{align}
    \widetilde\nstate_L = \sum_{\gamma\in\Gamma_{\mathcal{A}}}\coef(\gamma)\,Q_L(\gamma).
\end{align}
Consequently, the difference $\Delta\nstate_L := \nstate_L-\widetilde\nstate_L$ is the sum over the excluded (or ``killed'') paths:
\begin{equation}
\label{eq:delta-rho-paths}
    \Delta\nstate_L = \sum_{\gamma\notin\Gamma_{\mathcal{A}}} \coef(\gamma)\,Q_L(\gamma).
\end{equation}
The error in the expectation value of the observable $O$ is therefore:
\begin{equation}
\label{eq:error}
    \Tr(O\,\Delta\nstate_L)=\sum_{\gamma\notin\Gamma_{\mathcal{A}}} \coef(\gamma)\,\Tr\!\big(O\,Q_L(\gamma)\big) \, .
\end{equation}
Since $O$ is $\owt$-local, $O=\sum_{\wt(Q)\le \owt} o_Q Q$, the term $\Tr(O\,Q_L(\gamma))$ is non-zero only if $\wt(Q_L(\gamma))\le \owt$. Thus, we may restrict the sum to excluded paths that end with a weight $\le \owt$:
\begin{align}
    \Tr\!\big(O(\nstate_L-\widetilde\nstate_L)\big)
    =\sum_{\substack{\gamma\notin\Gamma_{\mathcal{A}} \\ \wt(Q_L(\gamma))\le \owt}} \coef(\gamma)\,\Tr\!\big(OQ_L(\gamma)\big)\, .
\end{align}
Using the bound $|\Tr(OQ)| \le 2^n|o_Q|$ for any Pauli string $Q$ and applying the triangle inequality, we obtain:
\begin{equation}
\label{eq:proof-abs}
    \big|\Tr\!\big(O(\nstate_L-\widetilde\nstate_L)\big)\big|
    \;\le\;
    2^n\|o\|_1\sum_{\substack{\gamma\notin\Gamma_{\mathcal{A}} \\ \wt(Q_L(\gamma))\le \owt}} |\coef(\gamma)|.
\end{equation}
Finally, taking the expectation over the random choice of Hamiltonian terms yields
\begin{equation}\label{eq:averageerrorbound-pathweight}
    \mathbb E\big[ |\Tr\!\big(O(\nstate_L-\widetilde\nstate_L))|\big]
    \leq 2^n\|o\|_1 \, \mathbb E\Bigg[ \sum_{\substack{\gamma\notin\Gamma_{\mathcal{A}} \\ \wt(Q_L(\gamma))\le \owt}} |\coef(\gamma)| \Bigg] \, .
\end{equation}
Noting that $|\coef(\gamma)| \le \cosh(\tau)^{L-r(\gamma)}\,\sinh(\tau)^{r(\gamma)}$ completes the proof.
\end{proof}

\section{Technical lemmas for controlling the normalization via the partition function}

In this section we derive a bound that can translate an error guarantee for the approximate un-normalized state $\tilde{\nstate}_L$ into a bound for the approximate normalized state $\tilde{\rho_L} = \frac{\tilde{\nstate_L}}{\Tr[\tilde{\nstate_L}]} $ .

\begin{theorem}[Expectation value stability under additive error]
\label{thm:exp-values}
Let $\nstate\succeq 0$ be a positive semidefinite operator on a $d$-dimensional Hilbert space with $Z:=\Tr(\nstate)>0$. Let $\tilde \nstate$ be an approximation of $\nstate$, and define $\tilde Z:=\Tr(\tilde \nstate)$. We denote the normalized states and the error operator as:
\begin{align}
\rho := \frac{\nstate}{Z},\qquad \tilde\rho := \frac{\tilde \nstate}{\tilde Z},\qquad \Delta := \nstate-\tilde \nstate.
\end{align}
Assume that for all Pauli operators $P$ with weight $\abs{P}\leq \owt$, the error satisfies $\abs{\Tr[\Delta P]}\le d\, \varepsilon$ for some $\varepsilon < Z/d$. Then, for all such $P$,
\begin{equation}
\label{eq:trace-distance-bound-eps-bis}
\abs{\Tr[P(\rho-\tilde\rho)]}
\;\le\;
\frac{2 d\varepsilon}{Z - d\varepsilon} \, .
\end{equation}
Moreover, if $\nstate$ satisfies $Z \ge d$, 
then for $\varepsilon < 1$:
\begin{equation}
\label{eq:trace-distance-bound-eps-simplified-bis}
\abs{\Tr[P(\rho-\tilde\rho)]}
\;\le\;
\frac{2\varepsilon}{1-\varepsilon} \, .
\end{equation}
\end{theorem}

\begin{proof}
We start from the identity:
\begin{align}
\rho-\tilde\rho
&=\frac{\nstate}{Z}-\frac{\tilde \nstate}{\tilde Z}
=\frac{\nstate-\tilde \nstate}{Z}+\tilde \nstate\left(\frac{1}{Z}-\frac{1}{\tilde Z}\right)
=\frac{\Delta}{Z} + \tilde \nstate \left( \frac{\tilde Z - Z}{Z\tilde Z} \right).
\label{eq:exact-id}
\end{align}
Multiplying by $P$ and taking the trace, we apply the triangle inequality:
\begin{equation}
\label{eq:decomp2-bis}
\abs{\Tr[P(\rho-\tilde\rho)]}
\le
\frac{\abs{\Tr[P\Delta]}}{Z}
+
\abs{\Tr[P\tilde \nstate]}\frac{|Z-\tilde Z|}{Z\,\tilde Z}.
\end{equation}

We now bound each term. By assumption, $\abs{\Tr[P\Delta]} \leq d \varepsilon$. Since the identity $I$ is a Pauli operator of weight 0, we also have $|Z-\tilde Z| = |\Tr(\Delta)| \le d\varepsilon$. Consequently, $\tilde Z \ge Z - d\varepsilon > 0$.

Next, we bound $\abs{\Tr[P \tilde \nstate]}$. Using $\tilde \nstate = \nstate - \Delta$, we have:
\begin{align}
\abs{\Tr[P \tilde \nstate]} \le \abs{\Tr[P \nstate]} + \abs{\Tr[P \Delta]}.
\end{align}
Since $\nstate \succeq 0$ and $\|P\|_\infty = 1$, we have $\abs{\Tr[P \nstate]} \le \Tr(\nstate) = Z$. Thus,
\begin{align}
\abs{\Tr[P \tilde \nstate]} \le Z + d\varepsilon.
\end{align}

Substituting these estimates into \eqref{eq:decomp2-bis} gives:
\begin{align*}
\abs{\Tr[P(\rho-\tilde\rho)]}
&\le
\frac{d\varepsilon}{Z}
+
(Z+d\varepsilon)\frac{d\varepsilon}{Z(Z-d\varepsilon)} \\
&= \frac{d\varepsilon}{Z} \left( 1 + \frac{Z+d\varepsilon}{Z-d\varepsilon} \right) \\
&= \frac{d\varepsilon}{Z} \left( \frac{Z - d\varepsilon + Z + d\varepsilon}{Z - d\varepsilon} \right) \\
&= \frac{2 d\varepsilon}{Z - d\varepsilon}.
\end{align*}
Finally, if $Z \ge d$, then $d/Z \le 1$, and substitution yields the simplified bound $\frac{2\varepsilon}{1-\varepsilon}$.

\end{proof}

\begin{theorem}\label{thm:exp-values-onenorm}
    Let $\nstate\succeq 0$ be a positive semidefinite operator on a $d$-dimensional Hilbert space with $Z:=\Tr(\nstate)>0$. Let $\tilde \nstate$ be an approximation of $\nstate$, and define $\tilde Z:=\Tr(\tilde \nstate)$. We denote the normalized states and the error operator as:
\begin{align}
\rho := \frac{\nstate}{Z},\qquad \tilde\rho := \frac{\tilde \nstate}{\tilde Z},\qquad \Delta := \nstate-\tilde \nstate.
\end{align}
Assume that the error satisfies $\|\Delta\|_1\le d\, \varepsilon$ for some $\varepsilon < Z/d$. Then,
\begin{equation}
\label{eq:trace-distance-bound-eps-bis-2}
\|\rho -\tilde\rho\|_1
\;\le\;
\frac{2 d\varepsilon}{Z - d\varepsilon} \, .
\end{equation}
Moreover, if $\nstate$ satisfies $Z \ge d$, 
then for $\varepsilon < 1$:
\begin{equation}
\label{eq:trace-distance-bound-eps-simplified-bis-2}
\|\rho -\tilde\rho\|_1
\;\le\;
\frac{2\varepsilon}{1-\varepsilon} \, .
\end{equation}
\end{theorem}

\begin{proof}
    The proof is the same as for Theorem~\ref{thm:exp-values} substituting the condition $|\Tr[P\Delta]|\leq d\varepsilon$ by $\|\Delta\|_1\leq d\varepsilon$.

\end{proof}

\section{Small angle truncation}\label{app:small_angle}

In this section we prove error bounds for small angle truncation, presented in Theorem~\ref{thm:qDRIFT_truncation-informal}. This proof follows a similar argument to the proof of Theorem 3 in~\cite{lerch2024efficient}. We start by proving a bound for the error when truncating unnormalized matrices in Theorem~\ref{thm:small_angle_trunc}, after applying $L$ steps of imaginary time evolution. To prove this we count the maximum number of paths that we are truncating, and we bound them using counting arguments. In Theorem~\ref{thm:qDRIFT_truncation}, we combine Theorems~\ref{thm:small_angle_trunc}, and \ref{thm:exp-values} to bound the truncation error for normalized states.

\subsection{Small angle truncation for unnormalized matrices}

\begin{theorem}
\label{thm:small_angle_trunc}

Let $\nstate_L$ be the unnormalized state obtained after applying $L$ steps of an imaginary time evolution $e^{-\tau P_t}(\cdot) e^{-\tau P_t}$, where $P_t$ is sampled randomly at each step $t$.
Let $\mathcal{A} \subseteq \mathcal{P}_n$ be all the paths that accumulate less than $k$ non-identity updates $\Gamma_{\mathcal{A}} = \{ \gamma : r(\gamma) \le k \}$. Let $\widetilde\nstate_L$ be the state produced by truncating any path $\gamma$ that contains more than $k$ non-identity updates (i.e., discarding $\gamma$ if $r(\gamma) > k$). 
    If $k > eL\tau$, then
    \begin{align}
        \|\nstate_L-\widetilde\nstate_L\|_1
        \leq 2^ne^{\tau L} \left(\frac{e L \tau}{k}\right)^k\,.
    \end{align}
\end{theorem}

\begin{proof}
    We can slightly modify Lemma~\ref{lemma:general_error}. Starting from Eq.~\eqref{eq:delta-rho-paths}, where $\Delta\nstate_L = \nstate_L-\widetilde\nstate_L$
    
    \begin{equation}
\label{eq:delta-rho-paths-2}
    \Delta\nstate_L = \sum_{\gamma\notin\Gamma_{\mathcal{A}}} \coef(\gamma)\,Q_L(\gamma).
\end{equation}
    where $\Gamma_{\mathcal{A}} = \{ \gamma : r(\gamma) \le k \}$. Thus we can write 
    \begin{align}
    \|\Delta\nstate_L\|_1 =&\|\sum_{\gamma : r(\gamma) > k} \coef(\gamma)\,Q_L(\gamma)\|_1\\
    \leq&  2^n \, \mathbb \sum_{\gamma : r(\gamma) > k}  \cosh(\tau)^{L-r(\gamma)}\,\sinh(\tau)^{r(\gamma)}  \, .
    \end{align}
    Where we used triangular inequality to upper-bound all the paths as the $\sinh, \cosh$ weight times the one norm of any Pauli string.  There are at most $\binom{L}{j}$ paths with exactly $j$ non-identity updates (corresponding to choosing $j$ steps to apply the $\sinh$ branch). Thus, we can bound the sum by:
    \begin{align}
        S_{k,L} &:= \sum_{\gamma : r(\gamma) > k} \cosh(\tau)^{L-r(\gamma)}\,\sinh(\tau)^{r(\gamma)} \nonumber \\
        &\leq \sum_{j=k+1}^L \binom{L}{j} \cosh^{L-j}(\tau)\,\sinh^j(\tau) \, .
    \end{align}
    We factor out the total weight $e^{\tau L} = (\cosh\tau + \sinh\tau)^L$. Defining $p := \sinh(\tau)/e^\tau = (1-e^{-2\tau})/2$, we can rewrite the terms as:
    \begin{align}
        \cosh^{L-j}(\tau)\sinh^j(\tau) = e^{\tau L} (1-p)^{L-j} p^j \, .
    \end{align}
    Substituting this back into the sum, we recognize the tail of a Binomial distribution:
    \begin{align}
        S_{k,L} \le e^{\tau L} \sum_{j=k+1}^L \binom{L}{j} (1-p)^{L-j} p^j 
        = e^{\tau L} \Pr[X > k], \quad \text{where } X \sim \mathrm{Bin}(L, p).
    \end{align}
    Using the standard Chernoff bound for the tail of a binomial distribution, $\Pr[X > k] \le \left(\frac{e L p}{k}\right)^k$ for $k > Lp$. 
    
    Since $1-e^{-x} \le x$ for $x \ge 0$, we have $p \le \tau$. Thus, for $k > eL\tau \ge eLp$, the bound simplifies to:
    \begin{align}
        S_{k,L} \le e^{\tau L} \left(\frac{e L \tau}{k}\right)^k.
    \end{align}
    Combining this with the prefactor $2^n\|o\|_1$ yields the claim
    \begin{align}
         \|(\nstate_L-\widetilde\nstate_L)\|_1
        \leq 2^ne^{\tau L} \left(\frac{e L \tau}{k}\right)^k
    \end{align}
\end{proof}

Now we proceed to use this Theorem to prove that a similar bound holds for normalized quantum states.

\subsection{Small angle truncation for normalized states}

\begin{theorem}[qDRIFT small angle truncation error]
\label{thm:qDRIFT_truncation}
    Consider the simulation of the thermal state starting from the identity $\1$. The algorithm performs $L$ steps of the map $\mathcal{E}_t(\cdot) = e^{-\tau P_t}(\cdot) e^{-\tau P_t}$, with step angle $\tau = \frac{\beta \Lambda}{2L}$, where $\Lambda = \sum_j |h_j|$ is the sum of Hamiltonian coefficients.
    Let $\tilde\rho$ be the approximate state obtained by truncating any Pauli path that accumulates more than $k$ non-identity updates (i.e., discarding paths with more than $k$ factors of $\sinh\tau$). Then
    \begin{align}
        \|\rho - \tilde\rho\|_1\in \mathcal{O}\left( \, e^{\beta\Lambda/2} \left(\frac{e\beta\Lambda}{2k}\right)^k \right) \, .
    \end{align}
\end{theorem}

\begin{proof}
    We first bound the error for the unnormalized state $\nstate_L$. Invoking Theorem~\ref{thm:small_angle_trunc} with the qDRIFT step size $\tau = \frac{\beta \Lambda}{2L}$, the expected additive error is:
    \begin{align}
        \mathbb E\big[ |\Tr\!\big(O(\nstate_L-\widetilde\nstate_L))|\big]
        &\leq 2^ne^{\tau L}\left(\frac{eL\tau}{k}\right)^k \nonumber \\
        &= 2^n e^{\beta\Lambda/2}\left(\frac{e\beta\Lambda}{2k}\right)^k \, .
        \label{eq:unnorm_bound_qDRIFT}
    \end{align}
    
    To recover the error for the physical, normalized state $\rho = \nstate_L/Z$, we apply the stability bound from Theorem~\ref{thm:exp-values}. We observe that for imaginary time evolution starting from $\1$, the partition function satisfies $Z = \Tr(\nstate_L) \ge \Tr(\1) = 2^n$ (see Lemma~\ref{lemma:lowerbound_z}).
    
    Let $\varepsilon_{unnorm}$ denote the RHS of Eq.~\eqref{eq:unnorm_bound_qDRIFT}. The normalized error parameter $\varepsilon$ in Theorem~\ref{thm:exp-values-onenorm} is given by $\varepsilon = \varepsilon_{unnorm} / d$, where $d=2^n$. The dimension factor $2^n$ strictly cancels out:
    \begin{equation}
        \varepsilon = e^{\beta\Lambda/2}\left(\frac{e\beta\Lambda}{2k}\right)^k.
    \end{equation}
    
    Provided that $k$ is sufficiently large such that $\varepsilon < 1$, Theorem~\ref{thm:exp-values} dictates that the normalized trace distance is bounded by $\frac{2\varepsilon}{1-\varepsilon}$. In the asymptotic limit, this is linear in $\varepsilon$:
    \begin{equation}
        \|\rho-\tilde\rho\|_1 \in \mathcal{O}(\varepsilon) = \mathcal{O}\left( e^{\beta\Lambda/2}\left(\frac{e\beta\Lambda}{2k}\right)^k \right) \, .
    \end{equation}
\end{proof}

\begin{lemma}\label{lemma:scaling_bound_small_angles}
    Assume the model in Theorem~\ref{thm:qDRIFT_truncation}. Let $\beta\Lambda =  \Theta(\log n)$, and choose $k =   c \beta\Lambda$. Then the error in Theorem~\ref{thm:qDRIFT_truncation} scales as
        \begin{equation}
        \|\rho-\tilde\rho\|_1 \in  c^{-\Theta(\log n)}  = \frac{1}{{\rm poly} (n)}
    \end{equation}
\end{lemma}

\begin{proof}
    We start with by substituting $k  = c \beta \Lambda$ in the bound of Theorem~\ref{thm:qDRIFT_truncation}
    \begin{align}\label{eq:torecover_in_sacling_bound_small_angles}
       \|\rho-\tilde\rho\|_1   \in \mathcal{O}\left( e^{\beta\Lambda/2}\left(\frac{e}{2 c}\right)^{c\beta\Lambda} \right)
    \end{align}

    We can simplify this by taking the coefficients inside of the order and massaging them to obtain
    \begin{align}
        e^{\beta\Lambda/2}\left(\frac{e}{2 c}\right)^{c\beta\Lambda} = \left(\frac{e^{c + 1/2}}{(2c)^c}\right)^{\beta \Lambda}\leq 1 \forall c\geq 3\,.
    \end{align}

    Therefore, if we recover Eq.~\eqref{eq:torecover_in_sacling_bound_small_angles}, we find
    \begin{align}
        \|\rho-\tilde\rho\|_1   \in c^{-\Theta(\beta\Lambda)}
    \end{align}
    where we used that $1/\sqrt{c}^c\geq \frac{e^{c + 1/2}}{(2c)^c}$ and if we choose $\beta\Lambda\in \Theta(\log n)$ we find that the error is polynomially suppressed in $c,n$.
\end{proof}

\section{Weight truncation}\label{app:weight_trunc}

In this section we prove error bounds for weight truncation, presented in Theorem~\ref{thm:qDRIFT_backflow-informal}. This proof follows a similar argument to the proof of Theorem 3 in~\cite{lerch2024efficient}, and the small angle truncation presented in the previous section. We start by proving Theorem~\ref{thm:theorem_bound_O}, a bound for the error when truncating unnormalized matrices, after applying $L$ steps of imaginary time evolution. Differently to the case in the previous appendix, here we consider what is the probability that an truncated path would go back to a Pauli string that could overlap with our observable. We consider this by analyzing the probability of going back to a Pauli with weight smaller or equal to the maximum weight of the observable. In Theorem~\ref{thm:qDRIFT_backflow}, we combine Theorems~\ref{thm:theorem_bound_O}, and \ref{thm:exp-values} to bound the truncation error for normalized states. The resulting bound is less general that the small angle truncation bound in virtue of applying only to the expectation values of low weight observables, but strictly due to the additional error suppression from the suppressed probability of backflow.

\begin{theorem}[Backflow Truncation Bound for Unnormalized Matrices]
\label{thm:theorem_bound_O}
    Consider an imaginary time evolution of $L$ steps with angle $\tau$ starting from the identity $\1$. Let $\widetilde\nstate_L$ be the unnormalized operator produced by truncating any Pauli string $Q$ with weight $\wt(Q) > k$.
    
    Let $\owt$ be the weight of the observable $O$. Define $m$ as the minimum number of weight-reducing updates required to transform a Pauli string of weight $> k$ to one of weight $\le \owt$. Let $\pbf$ be the maximum probability that a random update reduces the weight of a Pauli string (the ``backflow probability'').
    
    The expected error in the observable is bounded by:
    \begin{equation}
        \mathbb E\big[ |\Tr(O(\nstate_L-\widetilde\nstate_L))|\big] 
        \leq 2^n\|o\|_1 \exp\!\Big(L\tau + \tfrac13 L\tau^3\Big)\left(\frac{e L \tau \pbf}{m}\right)^{m} \, .
    \end{equation}
\end{theorem}

\begin{proof}
    We start from the general bound derived in Lemma~\ref{lemma:general_error} (Eq.~\ref{eq:averageerrorbound-pathweight}). Using the bound $|\coef(\gamma)| \le e^{L\tau^2/2}\tau^{r(\gamma)}$ (where $r(\gamma)$ is the number of non-identity updates), we have:
    \begin{equation}\label{eq:averageerrorbound-sinh}
        \mathbb E\big[ |\Tr\!\big(O(\nstate_L-\widetilde\nstate_L))|\big]
        \leq 2^n\|o\|_1 e^{\frac{1}{2} L\tau^2} \, \mathbb E\Bigg[ \sum_{\gamma \in \mathcal{K}} \tau^{r(\gamma)} \Bigg] \, ,
    \end{equation}
    where $\mathcal{K}$ is the set of ``killed'' paths: those that exceed weight $k$ at some point but end with weight $\le \owt$.
    
    Any such path must realize at least $m$ weight-decreasing updates to bridge the gap from weight $>k$ down to $\owt$. Let $D(\gamma)$ denote the number of weight decreases in path $\gamma$. The condition $\gamma \in \mathcal{K}$ implies $D(\gamma) \ge m$.
    
    We now bound the expected sum $S := \mathbb E\Big[ \sum_{\gamma \in \mathcal{K}} \tau^{r(\gamma)} \Big]$. We first sum over the total number of non-identity updates $r$ (from $m$ to $L$). For a fixed $r$, the number of decreases $D$ follows a distribution bounded by a Binomial distribution $B(r, \pbf)$. Thus, we sum over the number of backflow steps $j$ (from $m$ to $r$):
    \begin{equation}
        S \le \sum_{r=m}^{L} \binom{L}{r} \tau^{r} \sum_{j=m}^{r}\binom{r}{j}\pbf^{\,j} \, .
    \end{equation}
    We use the identity $\binom{L}{r}\binom{r}{j}=\binom{L}{j}\binom{L-j}{r-j}$ to rearrange the binomial coefficients. This allows us to swap the summation order to sum over $j$ first:
    \begin{align}
        S &\le \sum_{j=m}^{L} \binom{L}{j} \pbf^{\,j} \sum_{r=j}^{L} \binom{L-j}{r-j} \tau^{r} .
    \end{align}
    We re-index the inner sum by setting $\ell = r-j$. As $r$ goes from $j$ to $L$, $\ell$ goes from $0$ to $L-j$:
    \begin{align}
        \sum_{r=j}^{L} \binom{L-j}{r-j} \tau^{r} 
        &= \tau^j \sum_{\ell=0}^{L-j} \binom{L-j}{\ell} \tau^{\ell} \nonumber \\
        &= \tau^j (1+\tau)^{L-j},
    \end{align}
    where the last equality follows from the binomial theorem. Substituting this back into the expression for $S$:
    \begin{align}
        S &\le \sum_{j=m}^{L} \binom{L}{j} \pbf^{\,j} \tau^j (1+\tau)^{L-j} \nonumber \\
        &= \sum_{j=m}^{L} \binom{L}{j} (\tau \pbf)^{j} (1+\tau)^{L-j} \, .
    \end{align}
    We now apply the tail bound for the binomial sum: $\sum_{j=m}^L \binom{L}{j} a^j b^{L-j} \le b^L \left(\frac{e L (a/b)}{m}\right)^m$. Identifying $a = \tau \pbf$ and $b = 1+\tau$, we get:
    \begin{align}
        S &\le (1+\tau)^L \left( \frac{e L \tau \pbf}{m(1+\tau)} \right)^m .
    \end{align}
    Since $1+\tau > 1$, we can simplify the denominator to obtain the strict upper bound:
    \begin{align}
        S \le (1+\tau)^L \left( \frac{e L \tau \pbf}{m} \right)^m \, .
    \end{align}
    Finally, we combine this with the prefactor from Eq.~\eqref{eq:averageerrorbound-sinh}. The total error prefactor is $\exp(L\tau^2/2) (1+\tau)^L$. Taking the logarithm and using the Taylor inequality $\log(1+\tau) \le \tau - \frac{\tau^2}{2} + \frac{\tau^3}{3}$ (valid for $\tau > 0$):
    \begin{align}
        \frac{1}{2}L\tau^2 + L\log(1+\tau) 
        &\le \frac{1}{2}L\tau^2 + L\left(\tau - \frac{\tau^2}{2} + \frac{\tau^3}{3}\right) \nonumber \\
        &= L\tau + \frac{1}{3}L\tau^3.
    \end{align}
    Exponentiating this result gives the final bound:
    \begin{equation}
        \mathbb E\big[ \Tr(O(\nstate_L-\widetilde\nstate_L))\big] \leq 2^n\|o\|_1 \exp\!\Big(L\tau + \tfrac13 L\tau^3\Big)\left(\frac{eL\,\tau\,\pbf}{m}\right)^{m} \, .
    \end{equation}
\end{proof}

Now we proceed to use this Theorem to prove that a similar bound holds for normalized quantum states.

\subsection{Weight truncation for normalized states}

\begin{theorem}[qDRIFT Backflow Error]
\label{thm:qDRIFT_backflow} Consider the qDRIFT simulation of a thermal state starting from the identity $\1$. The algorithm applies $L$ steps of size $\tau = \frac{\beta \Lambda}{2L}$, where $\Lambda = \sum |h_i|$. Let $\tilde\rho$ be the state obtained by truncating any Pauli string with weight exceeding $k$. Let $\pbf(k)$ be the backflow probability at weight $k$, and let $m$ be the minimum number of steps required to reduce a Pauli string of weight $>k$ to one overlapping with the $\owt$-local observable $O$. 
    Provided that $m > \frac{e\beta\Lambda\pbf(k)}{2}$, the error in the expectation value is bounded by:
    \begin{align}
        \abs{\Tr[O(\rho-\tilde\rho)]} \in \mathcal{O}\left( \|o\|_1 \, e^{\beta\Lambda/2} \left(\frac{e\beta \Lambda\,\pbf(k) }{2m}\right)^{m} \right) \, .
    \end{align}
\end{theorem}

\begin{proof}
    We begin with the general backflow bound from Theorem~\ref{thm:theorem_bound_O}. For an unnormalized evolution of $L$ steps with angle $\tau$, the expected error is:
    \begin{align}
        \mathbb E\big[ \Tr(O(\nstate_L-\widetilde\nstate_L))\big] 
        \leq d\|o\|_1 \exp\!\Big(L\tau + \tfrac13 L\tau^3\Big)\left(\frac{eL\,\tau\,\pbf(k)}{m}\right)^{m}.
    \end{align}
    We substitute the qDRIFT parameters. The total rotation angle is fixed to $\beta\Lambda/2$, so we set $\tau = \frac{\beta \Lambda}{2L}$. 
    
    First, we simplify the pre-factor. As $L \to \infty$, the cubic term vanishes:
    \begin{equation}
        L\tau + \frac{1}{3}L\tau^3 = \frac{\beta\Lambda}{2} + \frac{1}{3}L\left(\frac{\beta\Lambda}{2L}\right)^3 
        = \frac{\beta\Lambda}{2} + \mathcal{O}(L^{-2})\,.
    \end{equation}
    Next, we substitute $\tau$ into the geometric base:
    \begin{equation}
        \frac{e L \tau \pbf(k)}{m} = \frac{e L (\frac{\beta\Lambda}{2L}) \pbf(k)}{m} = \frac{e \beta \Lambda \pbf(k)}{2m}.
    \end{equation}
    Combining these, the unnormalized error bound is:
    \begin{equation}
        \label{eq:unnorm_backflow}
        \mathbb E\big[ \abs{\Tr(O\Delta\nstate_L)} \big]\in \order{ d \|o\|_1 e^{\beta\Lambda/2} \left(\frac{e \beta \Lambda \pbf(k)}{2m}\right)^m}.
    \end{equation}

    We can use Markov's inequality to see that 
    \begin{align}\label{eq:markov_qdrift}
        \Pr(\abs{\Tr(O\Delta\nstate_L)} < a \Ebb\big[ \abs{\Tr(O\Delta\nstate_L)} \big]) = 1-\Pr( \abs{\Tr(O\Delta\nstate_L)} > a \Ebb\big[ \abs{\Tr(O\Delta\nstate_L)} \big])\geq 1-\frac{1}{a}
    \end{align}

    To obtain the error for the physical state $\rho = \nstate_L/Z$, we invoke Theorem~\ref{thm:exp-values}. We identify the normalized error parameter $\varepsilon$ by dividing Eq.~\eqref{eq:unnorm_backflow} by the dimension $d$:
    \begin{equation}
        \varepsilon = a\|o\|_1 e^{\beta\Lambda/2} \left(\frac{e \beta \Lambda \pbf(k)}{2m}\right)^m.
    \end{equation}
    with probability at worst $1-1/a$, where we used Eq.~\eqref{eq:markov_qdrift}.
    Assuming the truncation threshold $k$ is large enough such that $\varepsilon < 1$, the stability theorem guarantees that the normalized error scales as $\mathcal{O}(\varepsilon)$. Thus obtaining
        \begin{align}
        \abs{\Tr[O(\rho-\tilde\rho)]} \in \mathcal{O}\left( \|o\|_1 \, e^{\beta\Lambda/2} \left(\frac{e\beta \Lambda\,\pbf(k) }{2m}\right)^{m} \right) \, .
    \end{align}
    
    Finally, we note that $m$ is determined by the locality of the Hamiltonian terms. If the maximum weight reduction per step is $\delta$ (e.g., $\delta=2$ for commutators of weight-2 Paulis), then $m = \lceil (k-\owt)/\delta \rceil$. This completes the proof.
\end{proof}

\begin{lemma}\label{lemma:scaling_bound_wbp}
    Assume the model in Theorem~\ref{thm:backflow_general}. Let $\beta\Lambda\in\order{\log(n)}$, $m\in\Theta(k)$, and $\|o\|_1\in\poly(n)$ then the error scales as
    \begin{align}
        \abs{\Tr[O(\rho-\tilde\rho)]} \in \, \left(\frac{\log(n)}{n}\right)^{\Theta(k)} \, .
    \end{align}
    for any of the models of $\pbf(w)$ considered in Appendix~\ref{app:pbf}.
\end{lemma}

\begin{proof}
    We start with the bound in Theorem~\ref{thm:backflow_general}. First we see that for any of the models of $\pbf(w)$ considered in Appendix~\ref{app:pbf} are suppressed either quadratically or linearly with the number of qubits $n$, i.e. $\pbf(w)\in\order{\frac{\weight}{n}}$. Thus, $\pbf(k) \beta\Lambda\in\order{\frac{k\log n}{n}}$. If we substitute this along side with $\beta\Lambda\in\order{\log(n)}, m\in\Theta(k)$, and $\|o\|_1\in\poly(n)$ we obtain
    \begin{align}
        \abs{\Tr[O(\rho-\tilde\rho)]} \in \mathcal{O}\left( {\poly }(n) \left(\frac{ \log n }{ n}\right)^{k} \right) \, .
    \end{align}
    which for sufficiently large $k$ (but $k\in\order{1}$) we can further simplify the expression to
    \begin{align}
        \abs{\Tr[O(\rho-\tilde\rho)]} \in \,  \left(\frac{ \log n }{n}\right)^{\Theta(k)}  \, .
    \end{align}
\end{proof}

\section{Probabilities of Backflow}\label{app:pbf}

In this section, we derive upper bounds for the backflow probability $\pbf(\weight)$—the probability that a random update reduces the weight of a Pauli string—under different interaction geometries and Hamiltonian distributions. In particular, we analyze the bounds on the backflow probability $\pbf(\weight)$ under four distinct weight 2 Hamiltonian models, defined by the geometry of the qubit interactions and the probability distribution of the applied gates. A summary of the backflow probabilities is provided in Table~\ref{tab:backflow}. 

\begin{description}
    \item[All-to-All Uniform] A fully connected geometry where any two-qubit gate is applied to a pair $(i,j)$ sampled uniformly from all $\binom{n}{2}$ possible edges. This model assumes perfect uniformity in interaction strength across the entire system.
    
    \item[Nearest-Neighbor Uniform] A 1D chain geometry where interactions are restricted to adjacent qubits $(i, i+1)$. The specific edge is sampled uniformly from the $n$ available nearest-neighbor pairs, representing an ideal linear topology.
    
    \item[General Close-to-Uniform] A model relaxing the uniformity assumption. The probability $p(e)$ of interacting on an edge $e$ is arbitrary, provided it does not deviate excessively from the uniform distribution $u(e)$. The deviation is constrained by the max-divergence (see Definition~\ref{def:max_divergence}) $D_\infty(p\|u) \le \alpha$, which implies the dominance condition $p(e) \le e^\alpha u(e)$.
    
    \item[Random Hamiltonian] A specific stochastic instance of the general model. Here, the interaction weights for edges are drawn i.i.d. from a continuous uniform distribution $\mathrm{Unif}(-1,1)$. We establish that these systems behave as close-to-uniform models with high probability, specifically satisfying the bound with $e^\alpha \approx 4$. These Hamiltonians include, for example, Spin Glasses~\cite{amoruso2003scalings}.
\end{description}

\begin{table}[h]\label{tab:backflow}
    \centering
    \renewcommand{\arraystretch}{2.0} 
    \setlength{\tabcolsep}{15pt} 
    
    \caption{Summary of Backflow Probability Bounds $\pbf(\weight)$.}
    \label{tab:backflow_summary}
    \begin{tabular}{l l l l}
        \hline \hline
        \textbf{Model} & \textbf{Geometry} & \textbf{Distribution $p$} & \textbf{Backflow Bound} $\pbf(\weight)$ \\ 
        \hline
        
        \multicolumn{4}{l}{\textbf{All-to-All Interactions}} \\
        \hspace{3mm} Uniform & Complete & Uniform ($u$) & $\displaystyle \in\Theta\left(\frac{\weight^2}{n^2}\right)$ \\
        \hspace{3mm} General & Complete & $D_\infty(p\|u) \le \alpha$ & $\displaystyle \in\order{\frac{e^\alpha\weight^2}{n^2}}$ \\
        \hspace{3mm} Random & Complete & $w_e \sim \text{Unif}(-1,1)$ & $\displaystyle \in\order{\frac{\weight^2}{n^2}}$ with probability $1 - e^{-\Theta(n^2)}$ \\[1.5ex]
        \hline
        
        \multicolumn{4}{l}{\textbf{Nearest-Neighbor (NN) Interactions}} \\
        \hspace{3mm} Uniform & 1D Chain & Uniform ($u_{\text{NN}}$) & $\displaystyle \in\order{\frac{\weight}{n}}$ \\
        \hspace{3mm} General & 1D Chain & $D_\infty(p\|u_{\text{NN}}) \le \alpha$ & $\displaystyle \in\order{\frac{e^\alpha\weight}{n}}$ \\
        \hspace{3mm} Random & 1D Chain & $w_e \sim \text{Unif}(-1,1)$ & $\displaystyle \in\order{\frac{\weight}{n}}$ with probability $1 - e^{-\Theta(n)}$ \\[1.5ex]
        \hline \hline
    \end{tabular}

\end{table}

\subsection{All-to-all two-body interactions}

\begin{lemma}[Backflow Probability for All-to-All interactions]

\label{lemma:qbf_all_to_all}
    Consider a Pauli string $Q$ of weight $\wt(Q) = \weight$. If we apply a random two-qubit Pauli gate sampled uniformly from all $\binom{n}{2}$ pairs and all $9$ non-identity Pauli combinations, the probability of backflow is bounded by:
    \begin{align}
        \pbf(\weight) \in\Theta\left(\frac{\weight^2}{n^2}\right)
    \end{align}
\end{lemma}

\begin{proof}
    Let $S=\supp(Q)$ be the support of the Pauli string, with $|S|=\weight$. Suppose an edge $e=(i,j)$ is sampled uniformly at random from the $\binom{n}{2}$ possible pairs, and a two-body Pauli operator $P_{ij}$ is sampled uniformly from the 9 possibilities on that edge.
    
    To calculate the conditional probability $\pbf(\weight) = \Pr(\text{Decay} \mid \text{Commute})$, we first classify the interaction by the size of the overlap $r:=|\{i,j\}\cap S|$ and analyze the conditions for weight reduction and commutation in each case:

    \begin{enumerate}
        \item \textbf{Case $r=0$ (Disjoint):} The gate acts on two sites where $Q$ is identity. 
        \begin{itemize}
            \item \textbf{Commutation:} The operators always commute ($[P, Q] = 0$).
            \item \textbf{Weight Change:} The weight strictly increases ($\Delta \weight = +2$).
        \end{itemize}

        \item \textbf{Case $r=1$ (Partial Overlap):} The gate acts on one site in $S$ and one outside.
        \begin{itemize}
            \item \textbf{Commutation:} The operators commute if and only if the Pauli matrices match on the single overlapping site (probability $1/3$). Otherwise, they anticommute (probability $2/3$).
            \item \textbf{Weight Change:} If they commute (match), the site in $S$ becomes identity, but the site outside $S$ gains a Pauli. The net weight change is $\Delta \weight = -1 + 1 = 0$.
        \end{itemize}

        \item \textbf{Case $r=2$ (Full Overlap):} The gate acts on two sites within $S$.
        \begin{itemize}
            \item \textbf{Commutation:} The operators commute if they anticommute on an even number of sites (0 or 2). This happens with probability $5/9$. They anticommute with probability $4/9$.
            \item \textbf{Weight Change:} A weight decrease ($\Delta \weight = -2$) occurs only if the gate $P_{ij}$ is the exact inverse of $Q$ on both sites (e.g., $P_{ij} = Q_i \otimes Q_j$). This specific configuration implies commutation.
        \end{itemize}
    \end{enumerate}

    We now compute the joint probability of decay and commutation, and the total probability of commutation.
    
    \paragraph*{Numerator: $\Pr(\rm{Decay} \cap \rm{Commute})$}
    Weight decay occurs exclusively in the $r=2$ case and requires the specific Pauli choice that cancels $Q$.
    \begin{align}
        \Pr(\text{Decay} \cap \text{Commute}) &= \Pr(r=2) \times \Pr(\text{Cancel} \mid r=2) \\
        &= \frac{\binom{\weight}{2}}{\binom{n}{2}} \cdot \frac{1}{9} \, .\label{eq:pbf_unif_c}
    \end{align}
    Indeed, there are 9 possible Pauli combinations. Only 1 cancels the string $Q$ perfectly. This 1 combination inherently commutes.

    \paragraph*{Denominator: $\Pr(\rm{Commute})$}
    It is more convenient to calculate the probability of anticommutation and subtract it from 1. Anticommutation occurs in cases $r=1$ and $r=2$:
    \begin{align}
        \Pr(\text{Anticommute}) &= \Pr(r=1)\Pr(\text{Anti} \mid r=1) + \Pr(r=2)\Pr(\text{Anti} \mid r=2) \\
        &= \frac{\weight(n-\weight)}{\binom{n}{2}} \cdot \frac{2}{3} + \frac{\binom{\weight}{2}}{\binom{n}{2}} \cdot \frac{4}{9} \\
        &= \frac{1}{9\binom{n}{2}} \Big[ 6\weight(n-\weight) + 4\frac{\weight(\weight-1)}{2} \Big] \\
        &= \frac{4\weight(3n - 2\weight - 1)}{9\binom{n}{2}} \, .\label{eq:pbf_unif_anti}
    \end{align}
    Therefore, the probability of commuting is:
    \begin{align}
        \Pr(\text{Commute}) = 1 - \Pr(\text{Anticommute}) = \frac{9\binom{n}{2} - 4\weight(3n - 2\weight - 1)}{9\binom{n}{2}} \, .
    \end{align}

    \paragraph*{The Ratio}
    Finally, we divide the numerator by the denominator:
    \begin{align}
        \pbf(\weight) &= \frac{\Pr(\text{Decay} \cap \text{Commute})}{\Pr(\text{Commute})} \\
        &= \frac{\frac{\binom{\weight}{2}}{9\binom{n}{2}}}{\frac{9\binom{n}{2} - 4\weight(3n - 2\weight - 1)}{9\binom{n}{2}}} \\
        &= \frac{\binom{\weight}{2}}{9\binom{n}{2} - 4\weight(3n - 2\weight - 1)}\\
        &\in\Theta\left( \frac{\weight^2}{n^2} \right)\, .
    \end{align}
\end{proof}

\subsection{Nearest-neighbor two-body interactions}

\begin{lemma}[Backflow Bound for Nearest-Neighbor Gates]
\label{lemma:nn_backflow_conditional}
    The probability of backflow $\pbf(\weight)$ for a Pauli string $Q$ of weight $\weight$, under the application of random nearest-neighbor two-qubit gates on a chain of $n$ qubits, is upper-bounded by:
    \begin{align}
        \pbf(\weight)\in\order{\frac{\weight}{n}}
    \end{align}
\end{lemma}

\begin{proof}
    Let $S = \supp(Q)$ be the support of the Pauli string. We analyze the probability of backflow $\pbf(\weight) = \Pr(\rm{Decay}|\rm{Commute}) =\frac{\Pr(\text{Decay} \cap \text{Commute})}{\Pr(\text{Commute})}$ by bounding the numerator and denominator separately over all possible geometric configurations of $S$.

    \paragraph*{Maximizing the Numerator (${\rm Decay} \cap {\rm Commute}$)}
    Weight decay can only occur if the sampled edge $e=(i, i+1)$ has non-identity terms in both $i$ and $i+1$. 
    \begin{itemize}
        \item \textbf{Geometric Constraint:} The number of bulk edges, denoted $k_{\text{bulk}}$, is maximized when the qubits in $S$ form a single contiguous block. In this configuration, $k_{\text{bulk}} = \weight - 1$. 
        \item \textbf{Operator Constraint:} On a chosen bulk edge, there are 9 equiprobable two-body Pauli operators. Decay occurs only if the operator cancels $Q$ on both sites (e.g., $P = Q_i \otimes Q_{i+1}$). There is exactly 1 such operator out of 9. This operator automatically commutes with $Q$.
    \end{itemize}
    Thus, the probability of decay and commutation is bounded by:
    \begin{equation}
        \Pr(\text{Decay} \cap \text{Commute}) = \frac{k_{\text{bulk}}}{M} \cdot \frac{1}{9} \le \frac{\weight - 1}{9M}.\label{eq:qbp_nn_c}
    \end{equation}
    where M is the total number of edges possible ($M = n-1$ for open boundaries and $M = n$ for closed boundaries).

    \paragraph*{Minimizing the Denominator (Commute)}
    To obtain a robust upper bound on the ratio, we must consider the worst-case scenario for the denominator. Minimizing the commutation probability is equivalent to maximizing the anticommutation probability. 
    \begin{itemize}
        \item \textbf{Geometric Constraint:} Anticommutation occurs primarily on ``boundary'' edges (connecting $S$ to non-$S$) where the probability is $2/3$, versus in edges where there are two Pauli matrices, where it is $4/9$. Since $2/3 > 4/9$, anticommutation is maximized when the number of boundary edges is maximal. This occurs when the errors in $S$ are maximally spread out (isolated), creating up to $2\weight$ boundary edges.
        \item \textbf{Probability Calculation:}
        \begin{equation}\label{eq:qbp_nn_anti}
            \Pr(\text{Anticommute}) \le \frac{2\weight}{M} \cdot \frac{2}{3} = \frac{4\weight}{3M}.
        \end{equation}
    \end{itemize}
    Consequently, the probability of commuting is lower-bounded by:
    \begin{equation}
        \Pr(\text{Commute}) = 1 - \Pr(\text{Anticommute}) \ge 1 - \frac{4\weight}{3M} = \frac{3M - 4\weight}{3M}.
    \end{equation}

    \paragraph*{Bounding the Ratio}
    Combining the maximized numerator and the minimized denominator:
    \begin{align}
        \pbf(\weight) = \frac{\Pr(\text{Decay} \cap \text{Commute})}{\Pr(\text{Commute})} 
        &\le \frac{\frac{\weight - 1}{9M}}{\frac{3M - 4\weight}{3M}} \\
        &= \frac{\weight - 1}{9M} \cdot \frac{3M}{3M - 4\weight} \\
        &= \frac{\weight - 1}{3(3M - 4\weight)} \\
        &= \frac{\weight - 1}{9M - 12\weight} \, .
    \end{align}
    
    If we have open boundary conditions $M = n-1$ and without we have $M = n$, thus we can write
    \begin{align}
        \pbf(\weight)\in\order{\frac{\weight}{n}}
    \end{align}
\end{proof}

\subsection{General close-to-uniform Hamiltonians}

In this section, we relax the assumption of perfect uniformity in the sampling of Hamiltonian terms. We characterize the deviation from the uniform distribution using the max-divergence $D_\infty(p\|u_{\text{NN}})$ (see Definition~\ref{def:max_divergence}) and derive a robust upper bound on the backflow probability.

\begin{theorem}[Backflow for General Distributions]
\label{thm:backflow_general}
    Let $u$ be the uniform sampling distribution over unordered pairs of $n$ qubits. Let $p$ be an arbitrary sampling distribution over these pairs such that its max-divergence from uniform is bounded by $D_\infty(p\|u) \le \alpha$. Then, the probability of backflow $\pbf(\weight)$ under 2-local Pauli gates distributed according to $p$ is bounded by:
\begin{align}
        \pbf(\weight)\in \order{ \frac{e^{\alpha}\weight^2}{n^2}}
    \end{align}
\end{theorem}

\begin{proof}
    Let $u$ be the uniform distribution of sampling one two-body all-to-all Pauli matrix. Let $p$ be a probability distribution such that the max-divergence $D_\infty(p\|u) \le \alpha$ (see Definition~\ref{def:max_divergence}) implies that for any event $E$, $p(E) \le e^\alpha u(E)$. Then in the worst case we find 
    \begin{align}
        \pbf(\weight) = \Pr_{p}(\rm{Decay}|\rm{Commute}) =\frac{e^\alpha\Pr_u(\text{Decay} \cap \text{Commute})}{1-e^{\alpha}\Pr_u(\text{Commute})}
    \end{align}
    where $\Pr_p(\cdot)$ denotes the probability of an event according to distribution $p$ and, similarly, $\Pr_u(\cdot)$ corresponds to the uniform. Therefore, recovering the probabilities from Eqs.~(\ref{eq:pbf_unif_c},\ref{eq:pbf_unif_anti})
    \begin{align}
        \pbf(\weight) \leq &\frac{e^{\alpha}\binom{\weight}{2}}{9\binom{n}{2} - e^{\alpha}4\weight(3n - 2\weight - 1)}\\
        \in & \order{ \frac{e^{\alpha}\weight^2}{n^2}}
    \end{align}

\end{proof}

\begin{corollary}[Backflow for Random Hamiltonians]
    Let the weights of the Hamiltonian terms be chosen i.i.d.\ from $\mathrm{Unif}(-1,1)$, defining a sampling distribution $p$. With probability at least $1-e^{-\Theta(n^2)}$, the backflow probability is bounded by:
    \begin{align}
        \pbf(\weight)\in \order{ \frac{\weight^2}{n^2}}
    \end{align}

\end{corollary}

\begin{proof}
    The coefficient of the Hamiltonian being sampled i.i.d from ${\rm Unif}(-1,1)$, is equivalent to sampling them from a probability distribution ${\rm Unif}(0,1)$ since qDRIFT does not care about the sign. Then, immediately Proposition~\ref{prop:random_ham}, we know that with probability $\ge 1-e^{-M/8}$, the distribution $p$ satisfies $D_\infty(p\|u) \le \log 4$ (see Definition~\ref{def:max_divergence}).
    Substituting $\alpha = \log 4$ (so $e^\alpha = 4$) into Theorem~\ref{thm:backflow_general} yields the following result immediately
    \begin{align}
        \pbf(\weight) \leq \frac{4\binom{\weight}{2}}{9\binom{n}{2} - 16\weight(3n - 2\weight - 1)}\in \order{ \frac{\weight^2}{n^2}}
    \end{align}
\end{proof}

\subsection{Non-uniform Nearest-Neighbor Interactions}

We now extend the robustness analysis to nearest-neighbor geometries where the sampling of edges may not be perfectly uniform, such as in systems with disordered coupling strengths. To prove this we use the Max Divergence $D_\infty(p\|u_{\text{NN}})$, introduced in Definition~\ref{def:max_divergence}

\begin{theorem}[Backflow for General Nearest-Neighbor Distributions]
    Let $u_{\text{NN}}$ be the uniform distribution over the $n$ possible nearest-neighbor edges on a chain. Let $p$ be an arbitrary sampling distribution over these edges such that $D_\infty(p\|u_{\text{NN}}) \le \alpha$. 
    
    Then, the probability of backflow $\pbf(\weight)$ for a Pauli string of weight $\weight$ is bounded by:
\begin{align}
    \pbf(\weight)\in\order{\frac{e^\alpha\weight}{n}}
    \end{align}
\end{theorem}

\begin{proof}
    Let $u$ be the uniform distribution of sampling one two-body all-to-all Pauli matrix. Let $p$ be a palisaded distribution such that the max-divergence $D_\infty(p\|u) \le \alpha$ (see Definition~\ref{def:max_divergence}) implies that for any event $E$, $p(E) \le e^\alpha u(E)$. Then in the worst case we find 
    \begin{align}
        \pbf(\weight) = \Pr_{p}(\rm{Decay}|\rm{Commute}) =\frac{e^\alpha\Pr_u(\text{Decay} \cap \text{Commute})}{1-e^{\alpha}\Pr_u(\text{Commute})}
    \end{align}
    where $\Pr_p(\cdot)$ denotes the provability of an event according to distribution $p$ and, similarly, $\Pr_u(\cdot)$ corresponds to the uniform. Therefore, recovering the probabilities from Eqs.~(\ref{eq:qbp_nn_c},\ref{eq:qbp_nn_anti})
    \begin{align}
        \pbf(\weight) \leq  \frac{e^\alpha(\weight - 1)}{9M - e^{\alpha}12\weight} \, .
    \end{align}

    If we have open boundary conditions $M = n-1$ and without we have $M = n$, thus we can write
    \begin{align}
        \pbf(\weight)\in\order{\frac{e^\alpha\weight}{n}}
    \end{align}

\end{proof}

\begin{corollary}[Backflow for Random Nearest-Neighbor Hamiltonians]
    Let the weights of the nearest-neighbor Hamiltonian terms be chosen i.i.d. from $\mathrm{Unif}(-1,1)$. With probability at least $1-e^{\Theta(n)}$, the backflow probability is bounded by:
   \begin{align}
        \pbf(w)\in \order{\frac{\weight}{n}}\, .
    \end{align}
\end{corollary}

\begin{proof}
    The coefficient of the Hamiltonian being sampled i.i.d from ${\rm Unif}(-1,1)$, is equivalent to sampling them from a probability distribution ${\rm Unif}(0,1)$ since qDRIFT does not care about the sign. Then, by Proposition~\ref{prop:random_ham}, substituting the number of edges $M=\Theta(n)$, the distribution $p$ satisfies $D_\infty(p\|u_{\text{NN}}) \le \log 4$ with probability $\ge 1-e^{-\Theta(n)}$. Setting $e^\alpha = 4$ in the preceding theorem leads to
      \begin{align}
        \pbf(\weight) \leq  \frac{4(\weight - 1)}{9M - 48\weight} \in \order{\frac{\weight}{n}}\, .
    \end{align}
    
\end{proof}

{
\section{Extended weight-truncation analysis}\label{app:wt-trunc-trotter}

\subsection{Proof outline}
This section formalizes a simple intuition: for a \emph{high-temperature} Gibbs operator $e^{-\beta H}$, the overlap with a \emph{high-weight} Pauli string $Q$ is typically very small, so discarding high-weight Pauli components can be a controlled approximation.
While Section~\ref{app:weight_trunc} develops this idea for a specific simulation primitive (e.g.\ \emph{qDRIFT}), here we give an alternative analysis that applies more broadly to algorithms that build thermal operators by composing general maps, and we provide a bound specialized on 1st order Trotter formula.
Concretely, we consider a procedure that produces an (unnormalized) thermal operator by applying a sequence of maps to the identity:
\[
A_j \;:=\; \mathcal{M}_{\lfloor j,1\rfloor}(\mathbb{I})
\;=\; \mathcal{M}_j\circ \mathcal{M}_{j-1}\circ \cdots \circ \mathcal{M}_1(\mathbb{I}),
\qquad A_0 := \mathbb{I}.
\]

To keep the simulation classically tractable, we introduce a \emph{weight-truncated} surrogate in which we project onto Pauli strings of weight $<k$ after each step. Let $\Pi_{<k}$ denote the projector onto the span of Pauli strings of weight $<k$ (equivalently, $\Pi_{<k}$ discards all Pauli components of weight $\ge k$). We define the truncated forward evolution recursively by
\[
\widetilde A_0 := \mathbb{I},
\qquad
\widetilde A_j \;:=\; \Pi_{<k}\! \circ\mathcal{M}_j\left(\widetilde A_{j-1}\right),
\qquad (j=1,\dots,L),
\]
so that $\widetilde A_L$ is obtained from the same forward maps as $A_L$, but with truncation applied between layers.

Given an observable $O$, our estimator is $\Tr[O\widetilde A_L]$, intended to approximate $\Tr[OA_L]$.
It is convenient to evaluate the approximation error in the Heisenberg picture. Since $\Pi_{<k}$ is self-adjoint with respect to the Hilbert--Schmidt inner product, i.e.\ $\Pi_{<k}^\dag=\Pi_{<k}$, one has
\[
\Tr[O\widetilde A_L]
\;=\;
\Tr[\widetilde O_0],
\]
where the backward-truncated observables are defined by
\[
\widetilde O_L := O,
\qquad
\widetilde O_{j-1} \;:=\; \mathcal{M}_j^\dag\circ \Pi_{<k}\!\left(\widetilde O_j\right),
\qquad (j=L,L-1,\dots,1).
\]
Importantly, the algorithm itself is the forward construction of the truncated surrogate $\widetilde A_L$, and the backward recursion is used only as an \emph{analytical tool}, which allows to track the truncation error layer-by-layer via a telescoping identity that compares $\Tr[OA_L]$ and $\Tr[O\widetilde A_L]$.

The resulting error has two conceptually distinct sources that we bound separately.

\smallskip
\noindent
\emph{(i) Thermal suppression of high-weight Paulis.}
The first ingredient is a bound on Pauli expectation values taken on the partially evolved operator $\mathcal{M}_{\lfloor j,1 \rfloor}(\mathbb{I})$,
\[
\max_{Q:|Q|\ge k} 2^{-n}\,\left|\Tr(Q\mathcal{M}_{\lfloor j,1 \rfloor}(\mathbb{I}))\right|\;\le\; C_k.
\]
This formalizes idea that high-weight Paulis have small expectation in the (approximate) Gibbs operator when $\beta$ is small enough. 

\smallskip

\noindent
\emph{(ii) Controlling error accumulation under multiple truncations.}
The second ingredient is to bound the total mass of Paulis that can “leak” into the discarded sector during the backward propagation $O_j \mapsto \mathcal{M}_j^\dag(O_j)$. We quantify this using a layerwise Pauli-$\ell_1$ growth bound 
and a standard Pauli-path expansion, which together control the total mass that ever reaches weight $\ge k$. Combining this with a telescoping identity yields an error bound of the form
\[
2^{-n}|\Tr(O_0)-\Tr(OA_L)| \;\lesssim\; \,C_k\,e^{\delta L}\,\|O\|_{\mathrm{Pauli},1},
\]
making explicit how accuracy depends on thermal suppression of high-weight Paulis ($C_k$), and the cumulative growth of Pauli-$\ell_1$ mass under the adjoint layers ($e^{\delta L}$).

\subsection{Moment bound for Pauli expectation values: 1st order Trotter}
\begin{lemma}[Overlap bound for imaginary-time product-formula evolution]
\label{lem:trotter_overlap_partial}
Let $H=\sum_{a\in\mathcal T}\lambda_a Q_a$ be a Pauli Hamiltonian on $n$ qubits with
$|\lambda_a|\le 1$ and $|\supp(Q_a)|\le k$ for all $a\in\mathcal T$.
Fix an ordering $\pi=(a_0,\dots,a_{M-1})$ of $\mathcal T$ (so $M:=|\mathcal T|$) and a step size $\delta>0$.

For integers $p\ge 0$ and $j\in\{1,\dots,M-1\}$, define the \emph{partial} product-formula operators
\begin{align}
U_{\mathrm{right}}^{(p,j)}
\;:=\;
\Bigl(\prod_{m=1}^{M} e^{\delta\lambda_{a_m}Q_{a_m}}\Bigr)^{p}\;
\Bigl(\prod_{m=1}^{j} e^{\delta\lambda_{a_m}Q_{a_m}}\Bigr),
  \qquad  U_{\mathrm{left}}^{(p,j)}
\;:=\;
\Bigl(\prod_{m=j}^{1} e^{\delta\lambda_{a_m}Q_{a_m}}\Bigr)\;
\Bigl(\prod_{m=M}^{1} e^{\delta\lambda_{a_m}Q_{a_m}}\Bigr)^{p},
\end{align}

Consider the corresponding double-sided evolution from the identity
\[
 X_{p,j}
\;:=\;
U_{\mathrm{left}}^{(p,j)}\,\mathbb{I}\,U_{\mathrm{right}}^{(p,j)}
\;=\;
U_{\mathrm{left}}^{(p,j)}\,U_{\mathrm{right}}^{(p,j)}.
\]

Fix a Pauli $P$ with $R:=\supp(P)$, $w:=|R|$, and set $r:=\lceil w/k\rceil$.
Let
\[
\mathcal T_R:=\{a\in\mathcal T:\ \supp(Q_a)\cap R\neq\emptyset\},
\qquad B:=|\mathcal T_R|.
\]

Set $\beta = (p+1)\delta$.
Then for all $p\ge 0$ and $0\le j < M$,
\begin{equation}
\label{eq:trotter_overlap_partial_general}
2^{-n}\,\big|\Tr(P X_{p,j})\big|
\;\le\; \left(\frac{2e \beta B e^\delta}{r}\right)^r \,e^{2\beta M}.
\end{equation}

\end{lemma}

\begin{proof}
Each gate satisfies the expansion
\begin{equation}\label{eq:imag_gate_expand}
e^{\delta\lambda Q} = \cosh(\delta\lambda)\mathbb{I} + \sinh(\delta\lambda)Q.
\end{equation}
Fix an ordering $\pi = (a_0, \dots, a_{M-1})$ of the Hamiltonian terms, and let $L = pM + j$ for integers $p \ge 0$ and $0 \le j < M$. We define the squared partial product operator $X_{p,j} := (U_{\delta}^{(p,j)})^2$ by indexing the sequence of $2L$ gates appropriately.
By substituting \eqref{eq:imag_gate_expand} into the product and expanding, we obtain a sum over bitstrings $y \in \{0,1\}^{2L}$:
\begin{equation}\label{eq:X_single_sum_expand}
X_{p,j} = \sum_{z \in \{0,1\}^{2L}} \left( \prod_{t=0}^{2L-1} \cosh(\delta\lambda_{i_t})^{1-z_t} \sinh(\delta\lambda_{i_t})^{z_t} \right) \left( \prod_{t=0}^{2L-1} Q_{i_t}^{z_t} \right),
\end{equation}
where $Q^0 := \mathbb{I}$ and $Q^1 := Q$, and the operators in the final product are ordered by increasing $t$.

By Pauli orthogonality, the Pauli coefficient of $P$ in this expansion equals
$2^{-n}\Tr(P  X_{p,j})$, hence
\[
2^{-n}|\Tr(P X_{p,j})|
\le
\sum_{z:\ Q(z)=\omega P}\ \prod_{\ell=1}^{2L}
\bigl|\cosh(\delta\lambda_\ell)\bigr|^{1-z_\ell}\bigl|\sinh(\delta\lambda_\ell)\bigr|^{z_\ell},
\]
where $Q(z)\coloneqq \prod_{t=0}^{2L-1} Q_{i_t}^{z_t}$ denotes the resulting Pauli product.

Let $J_R$ be the set of factor positions $\ell\in\{1,\dots,2L\}$ whose Pauli term $Q_\ell$ acts
nontrivially on $R=\supp(P)$. 
As $X_{p,j}$ is obtained by performing $p$ full Trotter evolution plus an additional partial Trotter evolution containing $j$ distinct imaginary time rotations, and each imaginary time rotation contribute with 2 factor positions, we have
\[
|J_R|\le 2(p+1)B.
\]
If $Q(z)=\omega P$, then the union support of the selected Paulis must contain $R$.
Since each selected Pauli has weight at most $k$, at least $r=\lceil w/k\rceil$ of the selected
positions must come from $J_R$. Thus every contributing $z$ satisfies
$|\{\ell\in J_R:\ z_\ell=1\}|\ge r$.

Apply a union bound over subsets $S\subseteq J_R$ with $|S|=r$. For each fixed $S$,
dropping the constraint $Q(z)=\omega P$ only increases the sum, giving
\begin{align}
    2^{-n}|\Tr(P X_{p,j})|
\le&\sum_{\substack{S\subseteq J_R\\ |S| = r}}\,\sum_{\substack{\vec z = (z_1,\dots, z_{2L})
    \\ z_\ell = 1 \,{\rm if}\, \ell\in S\\ {z_\ell\in\{0,1\}\, \rm if}\, \ell\not\in S}} \prod_{\ell\in S}\sinh(\lambda_\ell\delta) \left[\prod_{\ell=0}^{2L-r}\cosh^{1-z_\ell}(\lambda_\ell\delta)\sinh^{z_\ell}(\lambda_\ell\delta)\right]
    \\ =& \sum_{\substack{S\subseteq J_R\\ |S| = r}}\, \sinh^r(\delta) \sum_{\substack{\vec y \in \{0,1\}^{2L-r}}}\left[\prod_{t=0}^{2L-r}\cosh^{1-y_\ell}(\delta)\sinh^{y_\ell}(\delta)\right]
\end{align}
where we used that $|\lambda_\ell|\le 1$ and therefore $|\sinh(\delta\lambda_\ell)|\le \sinh(\delta)$ and $|\cosh(\delta\lambda_\ell)|\le \cosh(\delta)$. Then, using the fact that $\sum_{\vec x\in\{0,1\}^m} \prod_{i=1}^m\alpha^{x_i} \beta^{1-x_i} = (\alpha + \beta)^{m}$, the sum in $\vec y$ can we rearranged in the following form
\begin{align}\label{eq:counting_all_the_terms}
2^{-n}|\Tr(P X_{p,j})|
\le
\sum_{S\subseteq J_R,\ |S|=r}\sinh(\delta)^r\bigl[\cosh(\delta)+\sinh(\delta)\bigr]^{2L-r}.
\end{align}

We can use that 
$\cosh(\delta)+\sinh(\delta)
= e^\delta$ to  obtain
\[
2^{-n}|\Tr(P X_{p,j})|
\le
\binom{|J_R|}{r}\,(\sinh\delta)^r\,e^{\delta(2L-r)}
\le
\binom{2(p+1)B}{r}\,(\sinh\delta)^r\,e^{2\delta(p+1)M},
\]
Note that here we are also considering the terms that have more than $r$ $\sinh$ paths, as in Eq.~\eqref{eq:counting_all_the_terms} we consider all the paths that have \textit{at least} $r$ $\sinh$ terms.

Using
$\binom{N}{r}\le (eN/r)^r$ and $\sinh\delta\le \delta e^\delta$, and setting $\beta = (p+1)\delta$, one finds that
\begin{align}
2^{-n}|\Tr(P X_{p,j})| \le
\binom{2(p+1)B}{r}\,(\sinh\delta)^r\,e^{2\delta(p+1)M}
\le \left(\frac{2e \beta B e^\delta}{r}\right)^r \,e^{2\beta M}.   
\end{align}
\end{proof}

\subsection{Weight-truncated Quantum Imaginary-Time Evolution}

Let $d:=2^n$ and let $\cP_n = \{I,X,Y,Z\}^{\otimes n}$ denote the $n$-qubit Pauli basis.
For a linear map $\cM=\cM_L\circ\cM_{L-1}\circ\cdots\circ\cM_1$ and indices $i\ge j$, write
\[
\cM_{\lfloor i,j\rfloor}:=\cM_i\circ\cM_{i-1}\circ\cdots\circ\cM_j .
\]
For any operator $X$ and Pauli $Q\in\cP_n$, define its (normalized) Pauli coefficient by
\[
\langle Q,X\rangle \;:=\;\frac{1}{d}\Tr(QX).
\]
Then $\Tr(QQ')=d\,\delta_{Q,Q'}$ and
\[
X=\sum_{Q\in\cP_n}\langle Q,X\rangle\,Q.
\]
Define the Pauli-$\ell_1$ norm by
\[
\|X\|_{\mathrm{Pauli},1}\;:=\;\sum_{Q\in\cP_n}\big|\langle Q,X\rangle\big|
\;=\;\frac{1}{d}\sum_{Q\in\cP_n}\big|\Tr(QX)\big|.
\]
For an integer cutoff $k\ge 1$, let $\Pi_{<k}$ and $\Pi_{\ge k}$ denote the projectors onto the span
of Pauli strings of weight $<k$ and $\ge k$, respectively:
\[
\Pi_{<k}(X):=\sum_{|Q|<k}\langle Q,X\rangle\,Q,
\qquad
\Pi_{\ge k}(X):=\sum_{|Q|\ge k}\langle Q,X\rangle\,Q.
\]

\medskip

\noindent\textbf{Assumptions.}
Let $H$ be a Hamiltonian and fix $\beta\ge 0$. We assume:
\begin{enumerate}
\item (\emph{Layerwise $\ell_1$-growth bound in the Pauli basis.})
For every Pauli $P\in\cP_n$ and every $j\in[L]$,
\[
\sum_{Q\in\cP_n}\left|\frac{1}{d}\Tr\!\bigl(Q\,\cM_j^\dag(P)\bigr)\right|\;\le\;e^\delta.
\]
Equivalently, for every operator $X$,
\begin{equation}\label{eq:pauli1_contraction_trace}
\|\cM_j^\dag(X)\|_{\mathrm{Pauli},1}\;\le\;e^\delta\,\|X\|_{\mathrm{Pauli},1}.
\end{equation}

\item (\emph{Thermal Pauli expectation bound at weight $\ge k$.})
Let the partially evolved operator at step $j$ be
\[
A_j:=\cM_{\lfloor j,1\rfloor}(\mathbb{I})\qquad (A_0:=\mathbb{I}).
\]
There exists a quantity $C_k$ such that
\begin{equation}\label{eq:Ck_assumption_trace}
\max_{Q\in\cP_n:\ |Q|\ge k}\ \frac{1}{d}\left|\Tr(Q A_j)\right|
\;\le\; C_k.
\end{equation}
\end{enumerate}

\medskip

\noindent\textbf{Weight-truncated propagation.}
Given an observable $O$, define recursively
\[
O_L:=O,
\qquad
\widetilde O_{j-1}:=\cM_j^\dag(O_j),
\qquad
O_{j-1}:=\Pi_{<k}(\widetilde O_{j-1})
\qquad (j=L,L-1,\dots,1).
\]
Define the discarded (high-weight) part at layer $j$:
\[
\Delta_j\;:=\;\Pi_{\ge k}(\widetilde O_{j-1})
\;=\;\widetilde O_{j-1}-O_{j-1}
\qquad (j=1,\dots,L).
\]
By construction,
\[
\|\Delta_j\|_{\mathrm{Pauli},1}
=\sum_{|Q|\ge k}\big|\langle Q,\Delta_j\rangle\big|
=\frac{1}{d}\sum_{|Q|\ge k}\big|\Tr(Q\,\Delta_j)\big|.
\]

\medskip

\noindent\textbf{Telescoping identity and truncation error.}
We compare $\Tr(O\,A_L)=\Tr(O_LA_L)$ with the truncated estimate $\Tr(O_0A_0)=\Tr(O_0)$.

\begin{lemma}[One-step telescoping]\label{lem:one_step_telescoping_trace}
For every $j\in\{1,\dots,L\}$,
\[
\Tr(O_jA_j)-\Tr(O_{j-1}A_{j-1})\;=\;\Tr(\Delta_j A_{j-1}).
\]
\end{lemma}

\begin{proof}
Using $A_j=\cM_j(A_{j-1})$ and duality,
\[
\Tr(O_jA_j)=\Tr\!\bigl(O_j\,\cM_j(A_{j-1})\bigr)=\Tr\!\bigl(\cM_j^\dag(O_j)\,A_{j-1}\bigr)
=\Tr(\widetilde O_{j-1}A_{j-1}).
\]
Since $\widetilde O_{j-1}=O_{j-1}+\Delta_j$, the claim follows.
\end{proof}

Summing Lemma~\ref{lem:one_step_telescoping_trace} over $j$ yields
\[
\big|\Tr(O_0)-\Tr(OA_L)\big|
\;\le\;\sum_{j=1}^L \big|\Tr(\Delta_jA_{j-1})\big|.
\]
Using the Pauli expansion $\Delta_j=\sum_{|Q|\ge k}\langle Q,\Delta_j\rangle Q$, we have
\[
\Tr(\Delta_jA_{j-1})
=\sum_{|Q|\ge k}\langle Q,\Delta_j\rangle\,\Tr(QA_{j-1}),
\]
so by triangle inequality,
\[
|\Tr(\Delta_jA_{j-1})|
\le
\Big(\sum_{|Q|\ge k}|\langle Q,\Delta_j\rangle|\Big)\,
\max_{|Q|\ge k}|\Tr(QA_{j-1})|
=
\|\Delta_j\|_{\mathrm{Pauli},1}\;
\max_{|Q|\ge k}|\Tr(QA_{j-1})|.
\]
Applying \eqref{eq:Ck_assumption_trace} gives $\max_{|Q|\ge k}|\Tr(QA_{j-1})|\le d\,C_k$, hence
\begin{equation}\label{eq:error_telescoping_bound_trace}
\big|\Tr(O_0)-\Tr(OA_L)\big|
\;\le\;
d\,C_k\sum_{j=1}^L \|\Delta_j\|_{\mathrm{Pauli},1}.
\end{equation}

\subsection{Bounding $\sum_{j=1}^L\|\Delta_j\|_{\mathrm{Pauli},1}$}

For each layer $j$ and Pauli $P$, expand $\cM_j^\dag(P)$ in the Pauli basis using trace coefficients:
\[
\cM_j^\dag(P)=\sum_{Q\in\cP_n} m^{(j)}(Q,P)\,Q,
\qquad
m^{(j)}(Q,P):=\langle Q,\cM_j^\dag(P)\rangle=\frac{1}{d}\Tr\!\bigl(Q\,\cM_j^\dag(P)\bigr).
\]
Also expand
\[
O=\sum_{P_L\in\cP_n}\langle P_L,O\rangle\,P_L
\qquad\text{where}\qquad
\langle P_L,O\rangle=\frac{1}{d}\Tr(P_LO).
\]
For a Pauli path $\gamma=(P_L,P_{L-1},\dots,P_0)\in\cP_n^{L+1}$ define its amplitude
\[
\Phi_\gamma
:=
\langle P_L,O\rangle\prod_{j=1}^L m^{(j)}(P_{j-1},P_j),
\]
so that
\[
\cM^\dag(O)
=\sum_{\gamma}\Phi_\gamma\,P_0.
\]
Define the path-$\ell_1$ mass
\[
\| \cM^\dag(O)\|_{\mathrm{paths},1}:=\sum_{\gamma}|\Phi_\gamma|.
\]

\begin{lemma}[Total path mass]\label{lem:path_mass_trace}
Under Assumption~1,
\[
\|\cM^\dag(O)\|_{\mathrm{paths},1}\;\le\; e^{\delta L}\,\|O\|_{\mathrm{Pauli},1}.
\]
\end{lemma}

\begin{proof}
Assumption~1 is exactly the statement that for every fixed $P$ and $j$,
\[
\sum_{Q\in\cP_n}|m^{(j)}(Q,P)|\le e^\delta.
\]
Therefore,
\begin{align*}
\sum_{\gamma}|\Phi_\gamma|
&=
\sum_{P_L}|\langle P_L,O\rangle|
\sum_{P_{L-1}}|m^{(L)}(P_{L-1},P_L)|
\cdots
\sum_{P_0}|m^{(1)}(P_0,P_1)| \\
&\le
\sum_{P_L}|\langle P_L,O\rangle|\,(e^\delta)^L
=
e^{\delta L}\,\|O\|_{\mathrm{Pauli},1}.
\end{align*}

\end{proof}

\begin{lemma}[Sum of truncation tails]\label{lem:sum_Delta_trace}
The truncation tails satisfy
\[
\sum_{j=1}^L \|\Delta_j\|_{\mathrm{Pauli},1}
\;\le\;
\|\cM^\dag(O)\|_{\mathrm{paths},1}.
\]
\end{lemma}

\begin{proof}
Partition the set of all paths $\gamma=(P_L,\dots,P_0)$ according to the \emph{first} layer at which
the path leaves the low-weight sector:
\[
\Gamma_0:=\{\gamma:\ |P_i|<k\ \ \forall i=0,1,\dots,L\},
\]
and for $j\in\{1,\dots,L\}$,
\[
\Gamma_j:=\{\gamma:\ |P_{j}|<k,\ |P_{j-1}|\ge k,\ \text{and }|P_i|<k\text{ for all }i=j,j+1,\dots,L\}.
\]
These sets are disjoint and cover all paths.

Fix $j\in\{1,\dots,L\}$ and $Q$ with $|Q|\ge k$. By construction of the recursion,
$\Delta_j=\Pi_{\ge k}(\cM_j^\dag(O_j))$ collects exactly those contributions whose first
high-weight Pauli occurs at stage $j-1$, hence
\[
\langle Q,\Delta_j\rangle
=\sum_{\gamma\in\Gamma_j:\ P_{j-1}=Q}\Phi_\gamma.
\]
Therefore, by triangle inequality,
\[
\|\Delta_j\|_{\mathrm{Pauli},1}
=\sum_{|Q|\ge k}|\langle Q,\Delta_j\rangle|
\le
\sum_{|Q|\ge k}\ \sum_{\gamma\in\Gamma_j:\ P_{j-1}=Q}|\Phi_\gamma|
=
\sum_{\gamma\in\Gamma_j}|\Phi_\gamma|.
\]
Summing over $j$ and using disjointness of the $\Gamma_j$ yields
\[
\sum_{j=1}^L\|\Delta_j\|_{\mathrm{Pauli},1}
\le
\sum_{j=1}^L\sum_{\gamma\in\Gamma_j}|\Phi_\gamma|
\le
\sum_{\gamma}|\Phi_\gamma|
=
\|\cM^\dag(O)\|_{\mathrm{paths},1}.
\]

\end{proof}

Combining \eqref{eq:error_telescoping_bound_trace} with Lemmas~\ref{lem:path_mass_trace}
and \ref{lem:sum_Delta_trace}, we obtain
\[
\big|\Tr(O_0)-\Tr(OA_L)\big|
\;\le\;
d\,C_k\,\|\cM^\dag(O)\|_{\mathrm{paths},1}
\;\le\;
d\,C_k\,e^{\delta L}\,\|O\|_{\mathrm{Pauli},1}.
\]
In particular, if $\|O\|_{\mathrm{Pauli},1}\le 1$ (e.g.\ $O$ is a Pauli observable), then
\[
\big|\Tr(O_0)-\Tr(OA_L)\big|
\;\le\;
d\,e^{\delta L}\,C_k.
\]

\subsection{Error bound: 1st order Trotter}

\begin{theorem}[Truncation error for 1st-order Trotter]
\label{thm:trunc_error_trotter_gle1}
Let $H=\sum_{a\in\mathcal T}\lambda_a Q_a$ be a Pauli Hamiltonian on $n$ qubits with
$|\lambda_a|\le 1$ and $|\supp(Q_a)|\le w$ for all $a$.
Fix an ordering $\pi=(a_1,\dots,a_M)$ of $\mathcal T$ and a step size $\delta>0$.
Let $p\ge 0$ and set $\beta:=(p+1)\delta$ and $L:=(p+1)M$.
Define the imaginary-time Trotter layers
\[
\cM_j(X):=e^{\delta\lambda_{a_j}Q_{a_j}}\,X\,e^{\delta\lambda_{a_j}Q_{a_j}},
\qquad j=1,\dots,L,
\]
with the periodic convention $a_{j+qM}:=a_j$ for $q=0,\dots,p$.
Let $A_j:=\cM_{\lfloor j,1\rfloor}(\mathbb{I})$ and let $O_0$ be the weight-$<k$ truncated backward
propagation output.

Assume $H$ has bounded degree $\ell$, i.e.\ each qubit appears in at most $\ell$ terms $Q_a$.
Define
\[
g:=2e\,\beta\,\ell\,w\,e^\delta.
\]
If $g\le 1$, then for every observable $O$,
\begin{equation}\label{eq:trunc_err_trotter_gle1_main}
\big|\Tr(O_0)-\Tr(OA_L)\big|
\;\le\;
d\,e^{2\delta L}\,e^{2\beta M}\,g^{\lceil k/w\rceil}\,\|O\|_{\mathrm{Pauli},1}
\;=\;
d\,e^{4\beta M}\,g^{\lceil k/w\rceil}\,\|O\|_{\mathrm{Pauli},1}.
\end{equation}
In particular, if $\|O\|_{\mathrm{Pauli},1}\le 1$ (e.g.\ $O$ is a Pauli), then
\[
\big|\Tr(O_0)-\Tr(OA_L)\big|
\;\le\;
d\,e^{4\beta M}\,g^{\lceil k/w\rceil}.
\]
\end{theorem}
\begin{proof}
We will apply the strategy described in the previous sections with $A_j=\cM_{\lfloor j,1\rfloor}(\mathbb{I})$.

\smallskip
\noindent\emph{Layerwise Pauli-$\ell_1$ growth.}
Fix a layer map $\cM_j(X)=GXG$ with $G=e^{\delta\lambda Q}$.
For any Pauli $P$, using $G=\cosh(\delta\lambda)I+\sinh(\delta\lambda)Q$,
\[
GPG=\cosh(\delta\lambda)^2P+\cosh(\delta\lambda)\sinh(\delta\lambda)(QP+PQ)+\sinh(\delta\lambda)^2QPQ.
\]
This is a linear combination of at most four Paulis, hence
\[
\|\cM_j^\dag(P)\|_{\mathrm{Pauli},1}=\|\cM_j(P)\|_{\mathrm{Pauli},1}
\le \big(|\cosh(\delta\lambda)|+|\sinh(\delta\lambda)|\big)^2
=e^{2|\delta\lambda|}\le e^{2\delta}.
\]
By linearity, $\|\cM_j^\dag(X)\|_{\mathrm{Pauli},1}\le e^{2\delta}\|X\|_{\mathrm{Pauli},1}$ for all $X$.
Therefore Lemma~\ref{lem:path_mass_trace} gives
\begin{equation}\label{eq:sum_delta_trotter}
\sum_{j=1}^L\|\Delta_j\|_{\mathrm{Pauli},1}\;\le\;e^{2\delta L}\,\|O\|_{\mathrm{Pauli},1}.
\end{equation}

\smallskip
\noindent\emph{Bounding $C_k$.}
Fix $j\in\{0,1,\dots,L\}$ and a Pauli $P$ with $|P|\ge k$.
Write $R:=\supp(P)$, $w_P:=|R|$, and set
\[
r:=\left\lceil\frac{w_P}{w}\right\rceil\ \ge\ \left\lceil\frac{k}{w}\right\rceil.
\]
There exist integers $p'\in\{0,\dots,p\}$ and $j'\in\{0,\dots,M-1\}$ such that
$A_j=X_{p',j'}$ in the notation of Lemma~\ref{lem:trotter_overlap_partial}, with
$\beta':=(p'+1)\delta\le\beta$.
Lemma~\ref{lem:trotter_overlap_partial} yields
\begin{equation}\label{eq:apply_overlap_trotter_Ck}
\frac{1}{d}\,|\Tr(PA_j)|
\;=\;2^{-n}|\Tr(PX_{p',j'})|
\;\le\;\left(\frac{2e\,\beta'\,B\,e^\delta}{r}\right)^r e^{2\beta' M}
\;\le\;\left(\frac{2e\,\beta\,B\,e^\delta}{r}\right)^r e^{2\beta M},
\end{equation}
where $B:=|\{a\in\mathcal T:\supp(Q_a)\cap R\neq\emptyset\}|$.

Using the bounded-degree assumption, each qubit in $R$ participates in at most $\ell$ terms, so
\[
B\;\le\;\ell\,|R|\;=\;\ell\,w_P\;\le\;\ell\,w\,r.
\]
Substituting into \eqref{eq:apply_overlap_trotter_Ck} gives
\[
\frac{1}{d}\,|\Tr(PA_j)|
\le
\Big(2e\,\beta\,\ell\,w\,e^\delta\Big)^r e^{2\beta M}
= g^{\,r}\,e^{2\beta M}.
\]
Since $g\le 1$ and $r\ge \lceil k/w\rceil$, we have $g^{r}\le g^{\lceil k/w\rceil}$, hence for all $j$
and all $|P|\ge k$,
\[
\frac{1}{d}\,|\Tr(PA_j)|
\le
e^{2\beta M}\,g^{\lceil k/w\rceil}.
\]
Therefore Assumption~\eqref{eq:Ck_assumption_trace} holds with
\begin{equation}\label{eq:Ck_value_trotter}
C_k := e^{2\beta M}\,g^{\lceil k/w\rceil}.
\end{equation}

\smallskip
\noindent\emph{Conclusion.}
Plugging \eqref{eq:sum_delta_trotter} and \eqref{eq:Ck_value_trotter} into
\eqref{eq:error_telescoping_bound_trace} yields
\[
|\Tr(O_0)-\Tr(OA_L)|
\le d\,C_k\sum_{j=1}^L\|\Delta_j\|_{\mathrm{Pauli},1}
\le d\,(e^{2\beta M}g^{\lceil k/w\rceil})\,(e^{2\delta L}\|O\|_{\mathrm{Pauli},1})
= d\,e^{4\beta M}\,g^{\lceil k/w\rceil}\,\|O\|_{\mathrm{Pauli},1},
\]
since $\delta L=\delta(p{+}1)M=\beta M$.
\end{proof}

\begin{corollary}[Truncation error for 1st-order Trotter]
\label{cor:trunc_error_trotter_normalized}
Let $H=\sum_{a\in\mathcal T}\lambda_a Q_a$ be a Pauli Hamiltonian on $n$ qubits with
$|\lambda_a|\le 1$ and $|\supp(Q_a)|\le w$ for all $a$.
Fix an ordering $\pi=(a_1,\dots,a_M)$ of $\mathcal T$ and a step size $\delta>0$.
Let $p\ge 0$ and set $\beta:=(p+1)\delta$ and $L:=(p+1)M$.
Define the imaginary-time Trotter layers
\[
\cM_j(X):=e^{\delta\lambda_{a_j}Q_{a_j}}\,X\,e^{\delta\lambda_{a_j}Q_{a_j}},
\qquad j=1,\dots,L,
\]
with the periodic convention $a_{j+qM}:=a_j$ for $q=0,\dots,p$.
Let $\rho:= \frac{\cM(\mathbb{I})}{\Tr[\cM(\mathbb{I})]}$ and let $\tilde \rho$ be the approximate version obtained by truncating Pauli operators with weight exceeding $k$.

Assume $H$ has bounded degree $\ell$, i.e.\ each qubit appears in at most $\ell$ terms $Q_a$.
Define
\[
g:=2e\,\beta\,\ell\,w\,e^\delta.
\]
If $g\le 1$, then for every observable $O$,
\begin{equation}\label{eq:trunc_err_trotter_gle1_proof}
\big|\Tr(O(\rho - \tilde \rho))\big|
\;\le\;
\, 4\,e^{4\beta M}\, g^{\lceil k/w\rceil}\,\|O\|_{\mathrm{Pauli},1}.
\end{equation}
\end{corollary}

\begin{proof}
Let $d=2^n$ and set
\[
O' \;:=\; \frac{O}{\|O\|_{\mathrm{Pauli},1}},
\qquad
\varepsilon \;:=\; e^{4\beta M}\, g^{\lceil k/w\rceil}.
\]
By Theorem~\ref{thm:trunc_error_trotter_gle1} (applied with observable $O'$), we have
\begin{equation}\label{eq:two_eps_bounds}
\frac1d\,\big|\Tr(A_L-\tilde A_L)\big|\le \varepsilon
\qquad\text{and}\qquad
\frac1d\,\big|\Tr\!\big(O'(A_L-\tilde A_L)\big)\big|\le \varepsilon.
\end{equation}
Recall that $\rho = A_L/\Tr(A_L)$ and $\tilde\rho=\tilde A_L/\Tr(\tilde A_L)$.
We invoke Theorem~\ref{thm:exp-values} with $A=A_L$, $\tilde A=\tilde A_L$, and $O'$.
Using \eqref{eq:two_eps_bounds}, it yields
\[
\big|\Tr\big(O'(\rho-\tilde\rho)\big)\big|
\;\le\;
\frac{2\varepsilon}{1-\varepsilon}.
\]
Therefore,
\[
\big|\Tr\big(O(\rho-\tilde\rho)\big)\big|
=
\|O\|_{\mathrm{Pauli},1}\,\big|\Tr\big(O'(\rho-\tilde\rho)\big)\big|
\;\le\;
\frac{2\varepsilon}{1-\varepsilon}\,\|O\|_{\mathrm{Pauli},1}.
\]

Finally, assume in addition that $\varepsilon\le \tfrac12$ (which can be ensured by choosing
$k$ large enough since $g<1$ implies $\varepsilon\to 0$ as $k\to\infty$).
Then
\[
\frac{2\varepsilon}{1-\varepsilon}
\;\le\;
\frac{2\varepsilon}{1-\tfrac12}
\;=\;
4\varepsilon,
\]
and hence
\[
\big|\Tr\big(O(\rho-\tilde\rho)\big)\big|
\;\le\;
4\,\varepsilon\,\|O\|_{\mathrm{Pauli},1}
\;=\;
4\,e^{4\beta M}\, g^{\lceil k/w\rceil}\,\|O\|_{\mathrm{Pauli},1}.
\]
This proves the claim.
\end{proof}

\section{From Pauli to Majorana propagation}\label{app:pauli_to_majorana}
In this section, we argue that the main theorems stated for Pauli propagation in the manuscript can be extended to the setting of Majorana propagation. We begin by recalling that, for imaginary-time evolution, conjugating a basis operator by an elementary gate has the same algebraic form in the Pauli and Majorana cases:

\begin{equation} e^{-\frac{\tau}{2}P}\,Q\,e^{-\frac{\tau}{2}P} = \begin{cases} Q, & [P,Q]\neq 0,\\[4pt] \cosh(\tau)\,Q - \sinh(\tau)\,P Q, & [P,Q]=0, \end{cases} \end{equation} \begin{equation}\label{eq:majorana_evolution} e^{-\frac{\tau}{2}M_{\mathbf b}}\,M_{\mathbf a}\,e^{-\frac{\tau}{2}M_{\mathbf b}}= \begin{cases} M_{\mathbf a}, & [M_{\mathbf b},M_{\mathbf a}]\neq 0,\\[4pt] \cosh(\tau)\,M_{\mathbf a} - \sinh(\tau)\,M_{\mathbf b}M_{\mathbf a}, & [M_{\mathbf b},M_{\mathbf a}]=0. \end{cases} \end{equation}

In particular, in both settings the operator branches only when it \emph{commutes} with the generator of the imaginary-time step; in that case, the two resulting branches are weighted by hyperbolic coefficients.

As a consequence, the small-coefficient truncation analysis (Theorem~\ref{thm:qDRIFT_truncation-informal}) extends essentially verbatim: it relies on combinatorial properties of propagation paths and on controlling the accumulation of $\sinh$ factors, and these features carry over directly to Majorana monomials via \eqref{eq:majorana_evolution}. In particular, consider the simulation of the imaginary time evolution for a Hamiltonian $H = \sum_{\vec b} h_{\vec b} M_{\vec b}$. Start from the identity $\1$, and apply $L$ steps of the map $\mathcal{E}_t(\cdot) = e^{-\tau M_{{\vec b}}/2}(\cdot) e^{-\tau M_{{\vec b}}/2}$, with step angle $\tau = \frac{\beta \Lambda}{L}$, where $\Lambda = \sum_{\vec b} |h_{\vec b}|$ is the sum of Hamiltonian coefficients. Then, if $\tilde\rho$ is the approximate state obtained by truncating any Majorana path that accumulates more than $k$ non-identity updates (i.e., discarding paths with more than $k$ factors of $\sinh\tau$), the $1$-norm error will scale as 
\begin{align}
    \norm{\rho - \tilde \rho}_1 \in \mathcal{O}\left(\, e^{\beta\Lambda/2} \left(\frac{e\beta\Lambda}{2k}\right)^k \right) \, .
\end{align}
Therefore, for any observable $O$ one has
    \begin{align}
        \abs{\Tr[O(\rho-\tilde\rho)]} \in \mathcal{O}\left( \|O\|_\infty \, e^{\beta\Lambda/2} \left(\frac{e\beta\Lambda}{2k}\right)^k \right) \, .
    \end{align}

The weight-truncation analysis (Theorems~\ref{thm:qDRIFT_backflow-informal} and~\ref{thm:trotter_weighttrunc_informal}) can be adapted in the same spirit. There, we use that non-identity Pauli operators are traceless and square to the identity, implying that the Hilbert--Schmidt inner product between two Pauli strings is nonzero only if they have identical support (equivalently, they are the same string up to phase). We then bound contributions by counting the number of propagation paths that lead to a prescribed support. The same reasoning applies to Majorana monomials: they are likewise traceless (for non-identity monomials) and square to the identity (up to a sign convention), and their Hilbert--Schmidt overlap is nonzero only when they act on the same set of Majorana modes (i.e., have the same \emph{mode support}). One can therefore again count the propagation paths that produce the relevant mode support and obtain analogous bounds. Before doing so, however, we need to establish that the Majorana length (denoted by $\ell$) plays the same role as the weight in this proof. Indeed, it is sufficient to bound the probability of backflow. Thus, we consider a qDRIFT simulation of a thermal state starting from the identity $\1$,  and let $O$ be a constant length $\ell$ observable. When applying $L$ steps of size $\tau = \frac{\beta \Lambda}{L}$, where $\Lambda = \sum |h_{\vec b}|$, the error in expected value for a state $\tilde\rho$ obtained by truncating any Majorana monomial with length exceeding $\ell_{\max}$ is
    \begin{align}
        \abs{\Tr[O(\rho-\tilde\rho)]} \in \mathcal{O}\left( \|o\|_1 \, e^{\beta\Lambda/2} \left(\frac{e\beta \Lambda\,\pbf(\ell_{\max}) }{2\ell_{\max}}\right)^{\ell_{\max}} \right) \, ,
    \end{align}
    for a sufficiently large but constant $\ell$. Note that $\pbf(\ell_{\max})$ is also the probability of backflow, but will be different for the Majoranas than it is for the Paulis.

As for \emph{first-order Trotterization}, consider a fermionic Hamiltonian $H$ of bounded degree $b$ that can be written as a sum of $M$ Majorana monomials, each of length at most $w$. Proceeding as in the proof of Theorem~\ref{thm:trotter_weighttrunc_informal}, one can show that, if $\tilde\rho$ is obtained by truncating all Majorana monomials of length larger than $\ell_{\max}$, then the resulting error in the expectation value of any observable $O$ satisfies
\begin{equation}
\Bigl|\Tr\!\bigl[O(\rho-\tilde\rho)\bigr]\Bigr|
\;\le\;
\|o\|_{1}\;
\exp\!\bigl(c_1\,\beta M\bigr)\;
\Bigl(c_2\,\beta\,b\,w\Bigr)^{\ell_{\max}/w},
\end{equation}
where $c_1,c_2>0$ are absolute constants.

\section{Additional numerical results}
In Fig.~\ref{fig:FH-energies} we supplement the results in Fig.~\ref{fig:fh-magnetization} with the estimated energy per imaginary time and the number of Majorana operators for various coefficient truncation thresholds. The 37-site data corresponds to the system studied in the main text, with an imaginary time step of $\tau=0.02$. It explains how an inverse temperature of $\beta=0.1$ is this challenging: We quickly generate billions of Majorana operators, which have a memory footprint of 74 qubits, i.e., 4 bits per site. With this and slightly over-allocated auxiliary data structures to compute more quickly, we reach the limit of $\sim500$GB memory computing nodes. Further optimizations are possible, but it does not appear particularly likely that even a factor of 2 would allow significantly further evolution. The 19-site system, equivalent to the 37 sites without the outermost hexagon layer, can be simulated further, though still notably less than the $J_1-J_2$ data in Fig.~\ref{fig:j1j2-numerics}.
\begin{figure}
    \centering
    \includegraphics[width=0.5\linewidth]{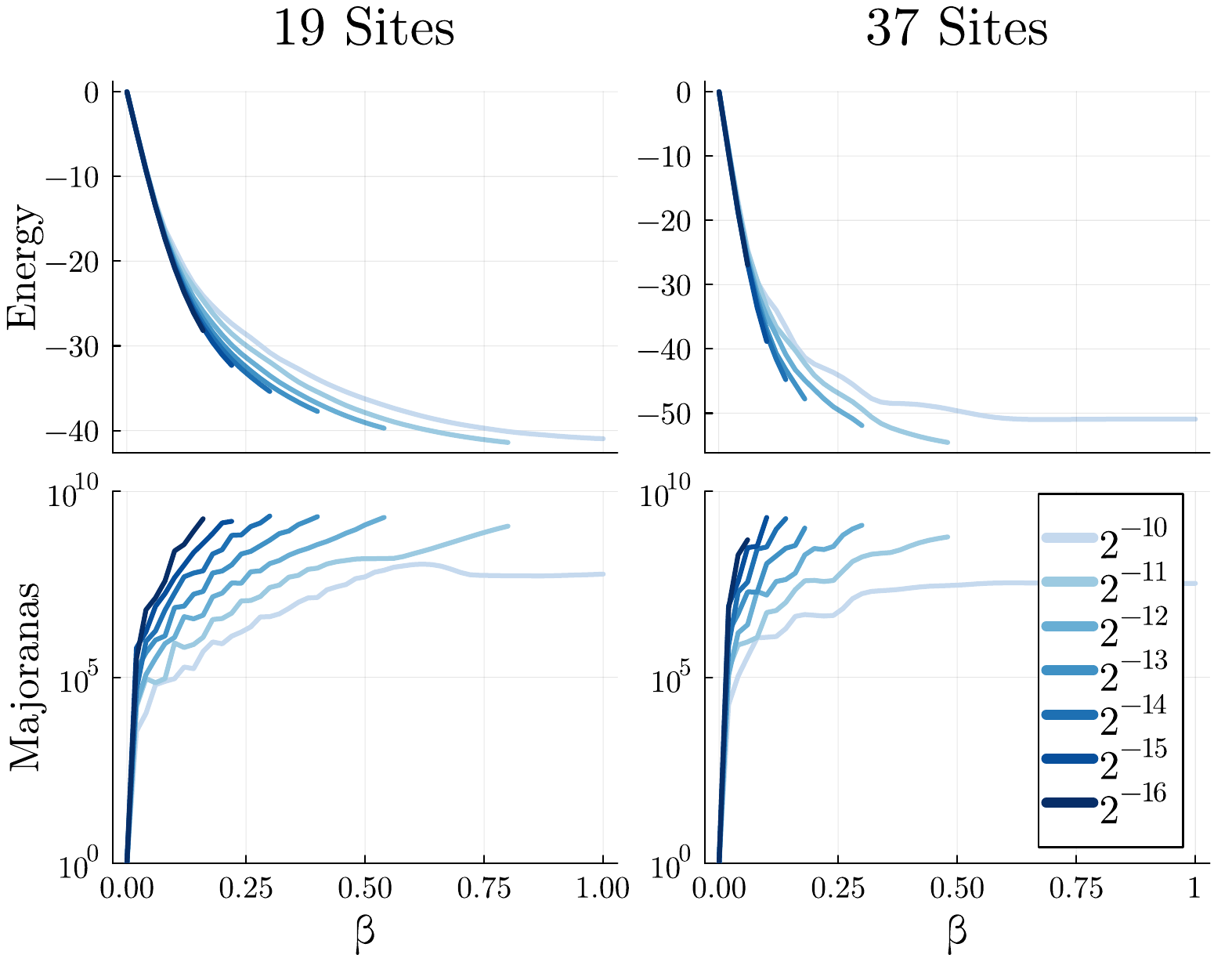}
    \caption{\textbf{Imaginary time evolution under the Fermi-Hubbard Hamiltonian.}  A look ``under the hood'' of the triangular lattice simulations from Figure 3. Here we track the energy (top) and the count of Majorana operators (bottom) as imaginary time progresses for 19-site and 37-site systems. The sharp rise in the bottom plots illustrates the ``operator-growth barrier": even at a high temperature of $\beta \approx 0.1$, the 37-site simulation generates billions of terms, saturating 500GB of memory. As before, the diverging lines show where stricter truncation thresholds cause the simulation to lose accuracy.}
    \label{fig:FH-energies}
\end{figure}

\end{document}